\documentclass[11pt]{article}

\usepackage[a4paper,margin=1in]{geometry}
\usepackage{amsmath,amssymb,amsthm,mathtools}
\usepackage{microtype}
\usepackage{xcolor}
\usepackage{booktabs}
\usepackage{enumitem}
\usepackage{authblk}
\usepackage{algorithm}
\usepackage[noend]{algpseudocode}
\usepackage{tikz}
\usetikzlibrary{arrows.meta,decorations.pathreplacing,positioning}
\usepackage[
  colorlinks=true,
  linkcolor=blue!55!black,
  citecolor=blue!55!black,
  urlcolor=blue!55!black
]{hyperref}

\newtheorem{theorem}{Theorem}[section]
\newtheorem{lemma}[theorem]{Lemma}
\newtheorem{proposition}[theorem]{Proposition}
\newtheorem{corollary}[theorem]{Corollary}
\theoremstyle{definition}

\newtheorem{openproblem}[theorem]{Open Problem}
\newtheorem{remark}[theorem]{Remark}

\newcommand{\ALG}{\operatorname{ALG}}
\newcommand{\OPT}{\operatorname{OPT}}
\newcommand{\LP}{\operatorname{LP}}
\newcommand{\E}{\mathbb{E}}
\newcommand{\Delay}{\operatorname{Delay}}
\newcommand{\Pages}{\mathcal{P}}
\newcommand{\Cache}{\mathcal{C}}
\newcommand{\Holes}{\mathcal{H}}
\newcommand{\Pending}{\mathcal{B}}

\title{Paging with Per-Replacement Maximum Delay}
\author{
Tianhang Lu,
Runtian Ren,
Shengcai Liu
}
\affil{
Guangdong Provincial Key Laboratory of Brain-Inspired Intelligent Computation,\\
Department of Computer Science and Engineering,\\
Southern University of Science and Technology, Shenzhen 518055, China\\
liusc3@sustech.edu.cn
}
\date{}

\begin{document}

\maketitle
\begin{abstract}
Classical paging serves every miss immediately.  We study paging with
per-replacement maximum delay, where loading a pending page costs one movement
plus the age of its oldest outstanding request and clears that page's entire
episode.  Equivalently, the holding rate is the number of pending pages.

For cache size $k$, threshold LRU is strictly $(5k+3)$-competitive, while a
randomized algorithm is strictly $5H_k$-competitive against an oblivious
adversary; classical constructions give matching $\Omega(k)$ and
$\Omega(H_k)$ orders.  Offline, we obtain an exact $O(nk)$ dynamic program
with one hole, an exact configuration algorithm for any fixed number of
holes, and a nonproactive polynomial-time $5$-approximation in general.
Physical farthest-next-use can nevertheless fail with only three pages.

For page-dependent fetch costs with spread
$\rho=w_{\max}/w_{\min}$, the exact fixed-number-of-holes algorithms persist.
We obtain a $(3\rho+2)$-approximation and online guarantees with multiplicative
factors $O(\rho k)$ deterministically and $O(\rho H_k)$ randomly against an
oblivious adversary, plus any additive term inherited from the corresponding
classical guarantee.  A strict $\Omega(\sqrt\rho)$ randomized lower bound
already holds for one cache slot.  Thus maximum delay preserves the unit-cost
competitive hierarchy but disrupts classical offline structure and makes
weighted timing spread-sensitive.
\end{abstract}


\section{Introduction}
\label{sec:introduction}

Paging is a standard model of online computation.  A cache holds $k$ pages,
requests arrive sequentially, and every miss is served immediately by loading
the requested page and evicting a cached page.  Three classical conclusions
shape the subject: deterministic algorithms have the characteristic
$\Theta(k)$ scale, randomized algorithms have the $\Theta(H_k)$ scale against
an oblivious adversary, and the offline optimum admits Belady's
farthest-next-use rule
\cite{belady1966study,sleator1985amortized,fiat1991competitive,
mcgeoch1991strongly}.

\paragraph{A motivating systems abstraction.}
Consider a cache manager between a small fast memory and a slower backing
store.  Moving an absent page into the fast tier consumes transfer bandwidth
and may evict a resident page.  Modern nonblocking memory systems can
\emph{combine} concurrent misses to the same block: later misses join one
pending entry, and a single returned block serves all of them
\cite{asiatici2022request,atre2020delayed}.  In the delayed-hits model, the
first miss starts a fetch of fixed duration.  We isolate a complementary
scheduling question for a system that may also control when a pending fill is
committed to a capacity-constrained cache.  An eager fill reduces the current
wait but may destroy a useful resident page and cause another transfer soon
afterwards; postponing it preserves current hits but increases the age of the
oldest request waiting for that page.

After normalizing one page transfer and replacement to unit cost, we model
this design as \emph{paging with per-replacement maximum delay}.  Requests
arrive on a continuous timeline and a miss may remain pending.  When a page
$p$ is eventually loaded, the system pays one unit of movement plus the age of
the oldest outstanding request to $p$, and all requests waiting for that page
are cleared together.  The delay term should therefore be read as the tail
waiting time of a page-level fill obligation, not as the sum or average of
individual request latencies.
The unit-cost model isolates the timing effect.  Section~\ref{sec:weighted}
later allows page-dependent transfer costs, as arise from heterogeneous
retrieval paths, remote-storage tiers, or per-page processing costs under a
fixed-slot abstraction.

The delay has a useful exact representation.  If $\Pending(t)$ is the set of
pages that are both missing and pending at time $t$, then
\begin{equation}
  \label{eq:intro-holding}
  \Delay=\int_0^\infty |\Pending(t)|\,dt.
\end{equation}
Thus the objective tracks active missing-page obligations rather than raw
request mass.  After one request has made a page pending, further requests to
that page do not increase the holding rate.  Requests to different missing
pages remain separate obligations because clearing them requires separate
replacements.

This accounting differs sharply from request-additive paging with delay.  If
one thousand requests arrive for the same absent page before it is loaded, a
request-additive objective counts one thousand waiting requests.  Here they
form one pending page episode, so that page contributes exactly one to the
holding rate until it is loaded.  The difference is therefore in what creates
urgency, not merely in a constant factor.

Waiting creates two coupled decisions absent from classical paging.  The
algorithm must decide \emph{when} to end a pending episode as well as which
cached page to evict.  Moreover, an eviction determines which future request
will become pending.  This makes the usual rent-or-buy analogy incomplete: a
pending page rents at unit rate and costs one to buy, but buying it changes the
state in which every later rental decision is made.  Randomization adds a
further difficulty.  If service times are derived from the physical cache,
the sequence seen by a delayed Marker rule may depend on its earlier victim
coins even when the raw timed input is oblivious.

Our main conclusion is asymmetric:
\begin{quote}
\emph{Maximum delay preserves paging's competitive hierarchy more robustly
than its classical offline optimality structure.}
\end{quote}
This statement describes the unit-cost model.  With unequal fetch costs,
the same structural tools survive, but arbitrary weight spread creates an
additional timing difficulty already with one cache slot.
In the unit-cost model, the online $\Theta(k)$ and $\Theta(H_k)$ orders
survive, but not by merely
slowing down classical algorithms.  A page-specific timer supplies the
missing temporal decision for LRU, while cache-independent windows give the
randomized paging subroutine an input-determined request sequence fixed before
its random choices.  Thus the classical
eviction rules remain useful only after a new timing rule removes the
dependence created by waiting.

Offline, one hole admits a compact exact dynamic program, and a fixed number
of holes admits an exact configuration dynamic program.  The same temporal
reduction used online also gives a polynomial-time $5$-approximation for an
unrestricted number of holes.  Nevertheless, the physical delayed optimum
need not obey farthest-next-use victim selection: waiting can make it optimal
to evict a page requested sooner.  Belady remains valid only inside the
ordinary virtual paging instance used by the approximation.  The exact
complexity for an unrestricted number of holes remains unresolved.

The weighted extension makes a second boundary precise.  Its structural
tools remain valid, but both our upper bounds and a one-slot lower bound now
depend on the weight spread.  Unlike classical weighted paging, the delayed
model must decide whether a cheap pending page should displace an expensive
resident page before that resident is requested again.

\paragraph{Scope of the contributions.}
The results have three layers.  Sections~\ref{sec:model}--\ref{sec:online-general-randomized}
give the core unit-cost paging theory.  Section~\ref{sec:weighted} is a full
weighted extension and identifies a strict spread-dependent timing barrier.
Appendix~\ref{sec:kserver-transfer} is only a general-metric outlook: it gives
a black-box transfer and isolates why a metric-independent theorem needs new
ideas.  Table~\ref{tab:intro-results} summarizes the first two layers.

One classical structural principle survives in an especially strong form.
Although the model permits proactive loading, a causal discrepancy projection
turns every schedule into a nonproactive one without increasing either its
movement or its delay.  Thus waiting is useful, but prefetching a page with no
pending obligation is not; the final cache being free is essential to this
statement.

These results have a simple explanation:
\begin{quote}
\emph{Page-wise maximum delay removes repeated requests from the urgency count
but leaves persistent cache state intact.}
\end{quote}

\paragraph{Techniques.}
The proofs use three main ideas.  First, the holding identity supports the
event dynamic programs and charges timer length directly to an offline
schedule without expanding repeated requests.

Second, two results use the same one-sided cache-discrepancy potential
\[
  \Phi(S,C)=|S\setminus C|.
\]
When a shadow replacement costs $f$ and a physical repair costs $g$, the
local inequality
\[
  g+\Delta\Phi\le f
\]
pays for the repair.  Applied to an arbitrary schedule, this is the causal
lazy projection: batch-time repairs reduce $\Phi$ for free and proactive
movement is never needed.  Applied after temporal aggregation, the same
potential implements a classical sequence of shadow-cache states using only
pending-page physical replacements.  Weighting each term of $\Phi$ by its
page's fetch cost gives the weighted form.

Third, the deterministic and randomized algorithms use deliberately
different timing rules.  Threshold LRU keeps the physical coupling and uses
OPT phases to charge activations, timers, and quarantined exceptions.  The
randomized algorithm instead removes cache dependence before choosing victims:
\[
  \text{raw input}
  \longrightarrow \sigma_\theta
  \longrightarrow \text{shadow paging}
  \longrightarrow \text{physical projection}.
\]
The two comparison inequalities
\[
  \ALG\le(1+\theta)F,
  \qquad
  \OPT_{\rm pg}(\sigma_\theta)
     \le \frac{2D}{\theta}+3M
\]
then have two uses: exact classical paging yields the offline
factor five, while the optimal classical randomized paging guarantee yields
the online factor $5H_k$.  Marker gives a simpler explicit $10H_k$
version.  For weights, page-dependent windows and a two-part charge for bad
representatives preserve the same architecture.

\subsection{Results}

Let $m$, $k$, and $r=m-k$ denote the page-universe size, the cache size, and
the number of holes, respectively.

\paragraph{Offline structure and algorithms.}
Every schedule has a cost-nonincreasing causal projection that loads only
currently pending pages.  With one hole,
the exact state consists of the hole identity and one pending bit, giving an
$O(nk)$ algorithm.  For arbitrary $r$, the exact state is a hole set together
with its pending subset, giving
\[
  O\!\left(n\binom mr2^r(1+rk)\right)
\]
time and hence polynomial time for fixed $r$.  We also characterize exactly
which states can be joined by a same-epoch lazy action word and quotient the
state space by the actual request support.

This exact state cannot currently be collapsed by choosing victims according
to their next use: a three-page instance makes physical farthest-next-use
strictly suboptimal.  For unrestricted $r$, however, fixed temporal windows
reduce the input to an ordinary paging sequence.  Solving that virtual
instance exactly and projecting it back gives a deterministic nonproactive
polynomial-time $5$-approximation.

\begin{table}[h]
  \centering
  \small
  \setlength{\tabcolsep}{6pt}
  \begin{tabular}{@{}p{0.22\linewidth}p{0.70\linewidth}@{}}
    \toprule
    \multicolumn{2}{@{}l}{\textbf{Offline structure and algorithms}} \\
    \midrule
    One hole
      & Exact $O(nk)$ dynamic program. \\
    Fixed $r$
      & Exact XP algorithm in $m^{O(r)}\operatorname{poly}(n)$ time. \\
    Arbitrary $r$
      & Exact complexity unresolved; deterministic nonproactive
        polynomial-time $5$-approximation. \\
    \midrule
    \multicolumn{2}{@{}l}{\textbf{Online competitive bounds}} \\
    \midrule
    One hole
      & Deterministic $(k+3)$ Timer-LRU; service-anchored Marker satisfies
        $\E[\ALG]\le 6H_k\OPT+2H_k+2$. \\
    General universe
      & Deterministic $(5k+3)$ Timer-LRU and randomized
        $5H_k$ AW-Partition against an oblivious adversary; both orders are
        asymptotically tight. \\
    \midrule
    \multicolumn{2}{@{}l}{\textbf{Weighted fetch costs}
      ($\rho=w_{\max}/w_{\min}$)} \\
    \midrule
    Fixed $r$, offline
      & Exact XP algorithm with the same state bound. \\
    General, offline
      & Deterministic nonproactive polynomial-time
        $(3\rho+2)$-approximation. \\
    General, online
      & Deterministic multiplicative factor $O(\rho k)$; randomized factor
        $O(\rho H_k)$ against an oblivious adversary.  A strict
        $\Omega(\sqrt\rho)$ randomized lower bound holds for $k=1$. \\
    \bottomrule
  \end{tabular}
  \caption{Summary of the results.  Here $r=m-k$ is the number of pages
  outside a cache of size $k$ in a universe of size $m$.  The general online
  bounds do not depend on $r$; one hole admits sharper constants.  Weighted
  bounds depend on the weight spread and recover constant-factor or
  classical-order guarantees when $\rho=O(1)$; the weighted online transfers
  inherit any additive term in the corresponding classical guarantee.}
  \label{tab:intro-results}
\end{table}

\paragraph{Online competitive bounds.}
For an arbitrary page universe, threshold LRU satisfies
\[
  N\le \frac{2D}{\theta}+(5k+1)M
\]
against every lazy comparator with movement $M$ and holding cost $D$, where
$N$ is the online service count.  Choosing $\theta=2/(5k+1)$ gives a strict
$(5k+3)$ ratio.  Running McGeoch and Sleator's partitioning algorithm in the
shadow cache after the same temporal aggregation satisfies
\[
  \E[\ALG]\le 5H_k\OPT
\]
against an oblivious adversary.  Classical lower-bound constructions transfer
to give $\Omega(k)$ and $\Omega(H_k)$, so both asymptotic orders are tight.
Using Marker instead gives a simpler explicit $10H_k$ version.
The lower bounds use requests separated by long time gaps.  Serving before
the next arrival then reproduces classical paging, while waiting until that
arrival already incurs a large holding cost.  Thus the classical
$\Omega(k)$ and $\Omega(H_k)$ examples remain valid in order.
When $r=1$, we additionally give a $(k+3)$-competitive threshold-LRU rule and
a service-anchored Marker bound of
$6H_k\OPT+2H_k+2$; these refinements appear in
Appendix~\ref{sec:online-one-hole}.

\paragraph{Weighted fetch costs.}
For page weights $w_p>0$, weighted causal laziness and the exact one-hole and
fixed-$r$ algorithms retain their unit-cost running times.  For arbitrary
$r$, page-dependent aggregation gives a deterministic polynomial-time
$(3\rho+2)$-approximation, a deterministic online guarantee with
multiplicative factor $O(\rho k)$, and a randomized guarantee with factor
$O(\rho H_k)$ against an oblivious adversary, where
$\rho=w_{\max}/w_{\min}$.  These upper bounds recover the unit-cost orders
when $\rho=O(1)$, with any additive term inherited from the underlying
classical weighted-paging algorithm.  A separate two-page construction gives a strict
$\Omega(\sqrt\rho)$ randomized lower bound already for $k=1$, showing that
arbitrary weight spread cannot disappear from a strict competitive guarantee.

\paragraph{LP relaxation and open questions.}
We formulate a polynomial marginal relaxation whose variables record clean
and pending hole mass.  It is not integral: a four-page instance gives an
integrality gap of at least $8/7$.  The $5$-approximation does not round this
LP, so improving the factor and determining the relaxation's full gap remain
separate questions.  Appendix~\ref{sec:approximation} gives the formulation
and certified gap.

\paragraph{Related work.}
Classical paging supplies three sharp benchmarks for the present study.
Belady's farthest-in-future rule is offline optimal
\cite{belady1966study}.  Sleator and Tarjan proved that LRU and FIFO are
$k$-competitive and that no deterministic online paging algorithm can improve
this factor \cite{sleator1985amortized}.  Fiat et al. introduced the randomized
marking framework, obtaining the harmonic $\Theta(H_k)$ competitive order
against an oblivious adversary \cite{fiat1991competitive}.  McGeoch and
Sleator's partitioning algorithm attained the optimal strict $H_k$ upper
bound for every fixed classical request sequence
\cite{mcgeoch1991strongly}.  Thus both online orders are classical, whereas
the temporal mechanisms used to recover them and the failure of physical
farthest-next-use are specific to maximum delay.

Classical weighted paging assigns each page $p$ a fetch cost $w_p$.
GreedyDual gives the optimal deterministic $k$ scale, and Bansal, Buchbinder,
and Naor give the optimal randomized $O(H_k)$ scale; these guarantees do
not depend on the ratio between the largest and smallest page weights
\cite{young1994kserver,bansal2007weighted}.  With maximum delay, the cache
manager must additionally decide whether a cheap pending page should displace
an expensive resident page before that resident is requested again.
Section~\ref{sec:weighted} shows that this timing choice forces weight-spread
dependence in strict competitiveness even when $k=1$.

Classical paging is also the uniform-metric special case of $k$-server.
On arbitrary metrics, the work-function algorithm is
$(2k-1)$-competitive under the standard convention
\cite{koutsoupias1995kserver}; a strict $(4k-2)$ form is also known
\cite{emek2010additive}.  Appendix~\ref{sec:kserver-transfer} gives a first
aspect-ratio-dependent transfer of any such classical guarantee to the direct
endpoint-batch version of our delayed model.  We keep this extension outside
the main development because removing the aspect ratio, and allowing several
moves to share one maximum-delay charge, require new ideas.

In weighted request-additive paging with delay, each outstanding request
contributes its own penalty.  Polylogarithmic online algorithms and
constant-factor offline approximations are known in that model
\cite{gupta2022caching}.  Our accounting replaces request mass by the number
of active missing-page obligations, so those results do not transfer.

Delayed hits are the closest systems-level comparison.  A miss starts a fetch
that completes after a fixed latency $Z$, and requests arriving for that page
in the meantime wait for the same return.  Atre et al. showed that classical
Belady need not minimize total request latency in this setting and developed
an offline latency-minimization algorithm \cite{atre2020delayed}.
Gurushankar et al. subsequently proved that LRU, and more generally marking
algorithms, are $O(Zk)$-competitive \cite{gurushankar2025latency}.  In our
model the service time is chosen by the algorithm, repeated requests add no
separate delay, and movement is an explicit part of the objective.  Thus the
two Belady failures concern different optimization problems, and neither the
offline algorithm nor the online guarantee transfers directly.

\begin{table}[h]
  \centering
  \small
  \setlength{\tabcolsep}{5pt}
  \newcommand{\modelrowdash}{%
    \noalign{\vskip2pt%
      \hbox to \dimexpr0.92\linewidth+4\tabcolsep\relax{%
        \color{black!35}%
        \leaders\hbox{\rule{2pt}{0.25pt}\hskip2pt}\hfill}%
      \vskip2pt}}
  \begin{tabular}{@{}p{0.21\linewidth}p{0.22\linewidth}p{0.49\linewidth}@{}}
    \toprule
    \raggedright Model
      & \raggedright Service timing
      & \raggedright Waiting and objective \tabularnewline
    \midrule
    \raggedright Classical paging
      & \raggedright Immediate on a miss
      & \raggedright Replacement count; no waiting term \tabularnewline
    \modelrowdash
    \raggedright Request-additive delay
      & \raggedright Chosen by the algorithm
      & \raggedright Movement plus a penalty for every pending request
        \tabularnewline
    \modelrowdash
    \raggedright Delayed hits
      & \raggedright Fixed $Z$ after the first miss
      & \raggedright Total latency of all requests \tabularnewline
    \modelrowdash
    \raggedright Bounded reordering
      & \raggedright Chosen within a prescribed window
      & \raggedright Paging cost subject to a hard service constraint
        \tabularnewline
    \modelrowdash
    \raggedright This work (uniform or weighted fetch costs)
      & \raggedright Chosen by the algorithm
      & \raggedright Page-load cost plus one oldest-age charge per page episode
        \tabularnewline
    \bottomrule
  \end{tabular}
  \caption{Nearby caching models differ in who controls service time and in
  what creates urgency.  All retain the persistent cache state that couples
  one replacement decision to later requests.}
  \label{tab:model-comparison}
\end{table}

Paging with bounded request reordering is the closest hard-constraint
counterpart.  Feder et al. allow each request to be delayed by at most a
prescribed number of time steps.  They give an offline dynamic program and
tight online bounds of $k+O(1)$ for deterministic algorithms and
$\Theta(\log k)$ for randomized algorithms \cite{feder2004combining}.
Albers later studied arbitrary cache sizes in a request-reordering model that
permits bypassing; the general offline problem is NP-hard, and small
constant-factor approximations follow from a reduction to batched service
schedules \cite{albers2010reordering}.  Our model replaces the prescribed
window by a soft episode-age cost and permits arbitrary postponement.

Reordering buffer management gives a broader comparison.  Colored items wait
for output in a capacity-limited buffer, and the objective counts color
changes.  Waiting is free but limited by buffer capacity, rather than priced
as in our model.  Its optimal randomized competitive ratio is
$\Theta(\log\log K)$ for buffer size $K$
\cite{adamaszek2011almost,avigdor2013rbm}, illustrating
that reordering does not in general preserve paging's $\Theta(H_k)$ order.

Bounded reordering and input time windows constrain when service is feasible,
whereas our delay is a soft cost and service remains feasible later.  In
particular, the fixed windows used in
Section~\ref{sec:online-general-randomized} are algorithmic aggregation
devices, not input deadlines.  On a uniform metric, request-additive
$k$-server with delay has a tight deterministic ratio $2k+1$
\cite{krnetic2020kserver}; its additive urgency accounting does not apply to
the present objective.

The objective is motivated by maximum-delay aggregation \cite{bhore2026general, lu2026mlamax}.  In multi-level
aggregation, younger requests may likewise join an older service without
increasing its delay, and a static submodular service structure supports a
polynomial offline dynamic program and DP-driven online timers
\cite{lu2026mlamax}.  Paging likewise counts one pending obligation per absent
page, but its state is not static: every replacement changes the cache from
which later decisions begin, coupling service timing with victim selection.
Consequently, the static-service dynamic-programming machinery does not carry
over directly; the present work instead requires configuration states,
victim-aware quarantine, and shadow-to-physical cache projections.  Online
service with delay also studies batching against waiting and identifies
reordering buffer management as a related problem \cite{azar2017osd}.
Methodologically, both analyses separate waiting obligations from a bounded
family paid by movement.  These are objective-level and proof-level analogies,
not reductions.

\subsection{Organization}

Section~\ref{sec:model} fixes the event semantics and holding identity.
Section~\ref{sec:offline} gives the exact offline structure and algorithms.
Section~\ref{sec:online-general} develops the deterministic online algorithm.
Section~\ref{sec:online-general-randomized} gives the common temporal
reduction, its offline approximation, and its randomized online
use.  Section~\ref{sec:weighted} extends the structural and temporal tools to
page-dependent fetch costs and proves the strict weight-spread lower bound.
Section~\ref{sec:conclusion} discusses what survives and what
changes.  Appendix~\ref{sec:online-one-hole} records the sharper one-hole
online bounds, and Appendix~\ref{sec:approximation} studies the separate
marginal LP.  Appendix~\ref{sec:kserver-transfer} records a first
general-metric $k$-server transfer and its single-scale limitation.

\section{Model and Structural Preliminaries}
\label{sec:model}

Let $\Pages$ denote a finite page set of size $m$, and assume $1\le k<m$; the
degenerate case $k=m$ has zero cost.  A cache configuration is a $k$-set
$\Cache\subseteq\Pages$; its complementary set
\[
  \Holes=\Pages\setminus\Cache,
  \qquad |\Holes|=r=m-k,
\]
is the set of \emph{holes}.  The initial cache $C_0$ is fixed and shared by the
online algorithm and the offline comparator.

All request and action times lie in $[0,\infty)$.  For every nonempty input,
translate time so that the first arrival epoch is zero.  A schedule starts
with the common initial cache at the event-order cut immediately before that
first complete batch; no action is permitted before this cut.  The empty
input has cost zero and will be omitted below.

Requests arrive at distinct epochs $0=t_1<t_2<\cdots<t_n$.  The atomic batch
$A_i\subseteq\Pages$ is revealed in full at time $t_i$.  Every request to a
cached page is satisfied immediately.  If a hole page $p$ is requested and
has no earlier outstanding request, it starts a pending episode at time
$a_p=t_i$.  Later requests to the same missing page join that episode and do
not reset $a_p$.
Multiple copies of one page at the same epoch may be collapsed without loss:
they all hit, or they start or join the same episode without changing its
oldest arrival time.  Thus representing each atomic batch as a set is
without loss of generality.

A replacement loads a hole page $p$ and evicts a cached page $q$.  The cache
and hole sets become
\[
  \Cache'=(\Cache\setminus\{q\})\cup\{p\},
  \qquad
  \Holes'=(\Holes\setminus\{p\})\cup\{q\}.
\]
If $p$ is pending and is loaded at time $s$, the replacement clears all
outstanding requests to $p$ and costs
\begin{equation}
  \label{eq:replacement-cost}
  1+(s-a_p).
\end{equation}
A replacement that loads a pending page will also be called a \emph{service}.
A replacement that loads a nonpending page is \emph{proactive} and costs
one.  A schedule with no proactive replacement is \emph{nonproactive} or
\emph{lazy}.  The final cache configuration is unrestricted, but every
pending request must eventually be served.

For every pending episode $e$, let $a_e$ be its activation time and let
$s_e\ge a_e$ be the time at which it is cleared.  We distinguish its
\emph{holding support}
\[
  \mathcal H_e=[a_e,s_e)
\]
from its \emph{closed service span}
\[
  \mathcal J_e=[a_e,s_e].
\]
Only $\mathcal H_e$ contributes to holding cost.  The closed span $\mathcal
J_e$ is used when a counting argument must retain activations and services at
its endpoints.  In particular, if $s_e=a_e$, then $\mathcal H_e$ is empty but
$\mathcal J_e=\{a_e\}$.  Both intervals have length $s_e-a_e$.

We write $M$ for the number of replacements and $D$ for the sum of the
page-episode delays.  The total objective is $M+D$.  Thus
Equation~\eqref{eq:replacement-cost} describes the combined contribution of
one lazy replacement, whereas a proactive replacement contributes only to
$M$.

\begin{remark}[Arrival-first convention]
If an arrival epoch coincides with one or more actions, the complete batch is
processed before any action.  A page that is cached during the batch is a
hit; if it is evicted later at the same timestamp, it becomes a clean hole.
Several zero-time replacements are ordered micro-actions, with no arrival
interleaved between them.
For a schedule $O$, write $C_O^-(t)$ for its cache at the event-order cut
immediately after the complete batch at real time $t$ and before its first
same-time action, and write $C_O^+(t)$ for the cache after all of its actions
at that timestamp.  A statement that an activation precedes a post-action
phase start refers to this event order; the activation and the phase start may
have the same real timestamp.
\end{remark}

Figure~\ref{fig:waiting-toy} isolates the new timing decision in the smallest
useful example.  Waiting can preserve a raw hit and strictly reduce total
cost even when the victim is unique.

\begin{figure}[t]
  \centering
  \begin{tikzpicture}[x=1cm,y=0.82cm,>=Latex,font=\small]
    \begin{scope}
      \node[font=\bfseries] at (3.10,3.45) {Immediate replacement};
      \node[draw,rounded corners,fill=blue!8,minimum width=1.15cm]
        at (0.70,2.55) {$\{c\}$};
      \node[draw,rounded corners,fill=blue!8,minimum width=1.15cm]
        at (2.90,2.55) {$\{a\}$};
      \node[draw,rounded corners,fill=blue!8,minimum width=1.15cm]
        at (5.50,2.55) {$\{c\}$};
      \draw[blue!65!black,thick,->] (1.28,2.55) -- (2.32,2.55);
      \draw[blue!65!black,thick,->] (3.48,2.55) -- (4.92,2.55);
      \node[blue!65!black,font=\scriptsize] at (2.90,1.95)
        {load $a$};
      \node[blue!65!black,font=\scriptsize] at (5.50,1.95)
        {load $c$};
      \draw[->,thick] (0.35,0.75) -- (6.00,0.75) node[right] {time};
      \draw (2.90,0.85) -- (2.90,0.65) node[below=3pt] {$a@0$};
      \draw (5.50,0.85) -- (5.50,0.65) node[below=3pt] {$c@1/2$};
      \node[align=center] at (3.10,-0.30)
        {two movements; delay $0$; cost $2$};
    \end{scope}

    \draw[gray!45] (7.20,-0.45) -- (7.20,3.75);

    \begin{scope}[xshift=7.55cm]
      \node[font=\bfseries] at (3.10,3.45) {Wait, then replace};
      \node[draw,rounded corners,fill=blue!8,minimum width=1.15cm]
        at (0.70,2.55) {$\{c\}$};
      \node[draw,rounded corners,fill=blue!8,minimum width=1.15cm]
        at (5.50,2.55) {$\{a\}$};
      \draw[blue!65!black,thick,->] (1.28,2.55) -- (4.92,2.55);
      \node[red!70!black,font=\scriptsize] at (2.90,1.95)
        {wait};
      \node[green!45!black,font=\scriptsize] at (5.50,1.95)
        {$c$ hits; load $a$};
      \draw[red!70!black,dashed,very thick] (2.90,1.45) -- (5.50,1.45)
        node[midway,below=2pt,font=\scriptsize] {$a$ pending};
      \draw[->,thick] (0.35,0.75) -- (6.00,0.75) node[right] {time};
      \draw (2.90,0.85) -- (2.90,0.65) node[below=3pt] {$a@0$};
      \draw (5.50,0.85) -- (5.50,0.65) node[below=3pt] {$c@1/2$};
      \node[align=center] at (3.10,-0.30)
        {one movement; delay $1/2$; cost $3/2$};
    \end{scope}
  \end{tikzpicture}
  \caption{Waiting can be strictly useful even with one cache slot.  The
  initial cache is $\{c\}$ and the requests are $a@0$ and $c@1/2$.
  Arrival-first processing lets the second request hit before the delayed
  load of $a$ evicts $c$.  This is a timing effect, not a victim-choice
  effect.}
  \label{fig:waiting-toy}
\end{figure}
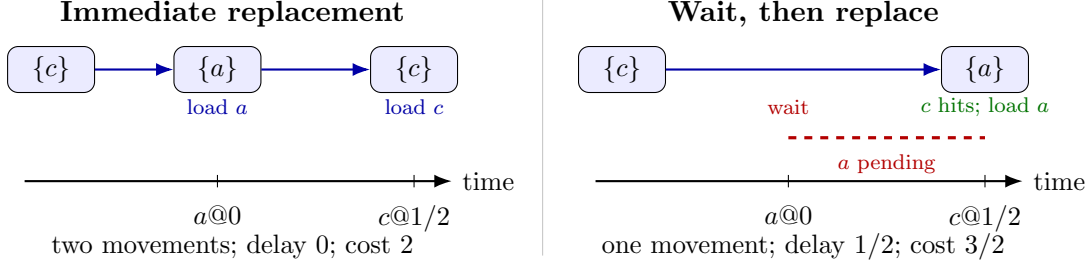

\paragraph{Competitive analysis.}
For a timed input $I$, let $\OPT(I)$ denote the minimum total replacement and
delay cost of an offline schedule with the common initial cache.  For
$c\ge1$ and $\beta\ge0$, a deterministic algorithm has a
$(c,\beta)$-competitive guarantee if
$\ALG(I)\le c\OPT(I)+\beta$ for every $I$.  A randomized algorithm has the
same guarantee against an oblivious adversary if
$\E[\ALG(I)]\le c\OPT(I)+\beta$ for every input fixed independently of its
random choices.  Here $\beta$ is independent of $I$.  When $\beta=0$, we
abbreviate $(c,0)$-competitive to $c$-competitive and may say
\emph{strictly $c$-competitive} to emphasize the absence of an additive
term.  Our upper bounds are strict unless an additive term is displayed
explicitly.  Inputs are revealed only
through their timed arrivals; there is no separate end-of-input notification.

\subsection{Holding-cost identity}

Let $\Pending(t)\subseteq\Holes(t)$ denote the pending holes at time $t$.

\begin{lemma}[Holding representation]
\label{lem:holding}
For every feasible schedule,
\[
  \Delay=\int_0^\infty |\Pending(t)|\,dt.
\]
\end{lemma}

\begin{proof}
Outside the finite set of activation and episode-clearing timestamps, a page
is pending exactly when
the time belongs to the holding support $\mathcal H_e$ of its current episode.
Thus, up to a measure-zero set,
\[
  |\Pending(t)|=\sum_e \mathbf 1_{\mathcal H_e}(t).
\]
Every episode contributes
\[
  \int_0^\infty \mathbf 1_{\mathcal H_e}(t)\,dt=s_e-a_e.
\]
Summing over the finitely many episodes proves the identity.  A zero-time
episode has empty holding support and contributes zero, as required.
\end{proof}

Thus an event-layer dynamic program need not store the ages $a_p$.  Past
waiting is already included in the state label, while the future incremental
cost over a gap of length $\Delta$ is $|\Pending|\Delta$.

\begin{lemma}[Arrival-layer normal form]
\label{lem:arrival-layer-normal-form}
Every finite feasible schedule can be transformed, without increasing cost
and without changing the cache seen by any arrival batch or the final cache,
so that all actions occur at post-batch cuts.  More precisely, for each
$i<n$, execute immediately after the batch at $t_i$ the complete action word
originally occurring in $[t_i,t_{i+1})$, retaining its chronological
micro-order; execute the terminal word $[t_n,\infty)$ after the batch at
$t_n$.
\end{lemma}

\begin{proof}
No request occurs after the batch at $t_i$ and before the batch at
$t_{i+1}$.  Moving the complete ordered action word for this half-open layer
to its left post-batch cut therefore crosses no arrival, preserves the cache
trajectory within the word, and leaves the cache seen by the next batch
unchanged.  Every page served by the word is pending after the left batch, so
each service age weakly decreases; movement and proactive costs are
unchanged.  The same argument applies to the terminal word after $t_n$.
\end{proof}

\subsection{Causal laziness}

\begin{lemma}[Causal lazy projection]
\label{lem:lazy-optimum}
When movement is uniform and the final cache configuration is free, every
feasible schedule $O$ admits a nonproactive schedule $L$ with
\[
  M_L\le M_O,
  \qquad
  D_L\le D_O.
\]
The projection is causal: if $O$ is generated online, $L$ can maintain $O$ as
an internal shadow execution and uses no future input.  Consequently,
proactive replacement cannot improve the offline optimum or the competitive
guarantee against an oblivious input.
\end{lemma}

\begin{proof}
Fix an arbitrary feasible schedule $O$, possibly proactive, and regard its
cache $S$ as a shadow trajectory.  We construct a lazy physical schedule with
cache $C$, starting from $C=S$, and use
\[
  \Phi=|S\setminus C|.
\]
Process raw batches and the ordered micro-actions of $O$ chronologically.
Fix a total order on pages once and for all; whenever several repairs or
victims are eligible, use this order.  Thus the projection introduces no
noncausal choice.

After a complete raw batch is evaluated against the pre-action caches, load
every page $p$ that is pending for the physical schedule and belongs to
$S\setminus C$.  Equal cache sizes give a page
$q\in C\setminus S$; replace $q$ by $p$.  This action is lazy and decreases
$\Phi$ by exactly one.  The batch itself did not change $S$, so the action
and the potential drop cancel.

Next process each replacement of $O$.  Its unit swap changes $S$ and can
increase $\Phi$ by at most one.  If its loaded page $p$ is physically pending
after the shadow update, then $p\in S\setminus C$; perform the same lazy
repair as above.  This repair decreases $\Phi$ by one.  Consequently, if
$f=1$ denotes the current shadow replacement and $g\in\{0,1\}$ the physical
repair, every shadow micro-action satisfies
\[
  g+\Delta\Phi\le f,
\]
while every batch-time repair satisfies $g+\Delta\Phi=0$.  The caches agree
initially and the final potential is nonnegative, so telescoping gives
\[
  M_L+\Phi_{\rm final}\le M_O+\Phi_{\rm initial}=M_O.
\]
In particular, the constructed schedule uses at most as many replacements as
$O$.

It remains to compare delay.  Consider a positive-length physical pending
episode $e=[a,s)$ for page $p$.  Had $p\in S$ when the batch at $a$ was
processed, the batch-time repair would have cleared $e$ at age zero.  Hence
$p\notin S$ at $a$, so the shadow execution has an active $p$-episode $e_O$
beginning at some $b\le a$.  No shadow action loads $p$ during $[a,s)$,
since its first such load would immediately trigger a physical repair.
Nor can a later batch-time repair be the first event clearing $e$: for $p$ to
enter $S$ after $a$, a preceding shadow load would already have cleared it.
Thus the shadow load at $s$ clears both episodes, so $e_O=[b,s)$ and
\[
  |e|=s-a\le s-b=|e_O|.
\]
Distinct positive-length physical episodes map to distinct shadow episodes,
because the common shadow load ends the associated episode before another
one can begin.  Zero-length repairs contribute no delay.  This injection
proves $D_L\le D_O$; feasibility of $O$ also implies that no physical episode
remains pending at the end.

Applying the construction to an optimum $O$ gives a lazy schedule of no
larger total cost.  Its final cache may differ from that of $O$, which is
permitted by the free-terminal-cache convention.  Every repair is determined
by the current raw batch, the current shadow micro-action, and the two current
caches, so the construction is causal.  For a randomized shadow algorithm
the inequalities hold pathwise for every fixed input $I$ and coin outcome
$\omega$, and hence also in expectation when $I$ is fixed obliviously.
\end{proof}

\begin{remark}[Boundary cases of the projection]
\label{rem:projection-boundaries}
The construction uses the arrival-first cut literally.  It tests the entire
batch against the caches that exist before any same-time action and then
processes shadow micro-actions sequentially.  Consequently, evicting a page
later at that timestamp cannot turn its earlier hit into a pending request.
A batch-time repair has age zero, and a shadow load of a page already in the
physical cache needs no repair.  After the last raw arrival, the projection
continues to follow the terminal shadow actions until all episodes are
cleared; this requires no advance notification that the input has ended.
The free-terminal-cache convention is used only to discard the nonnegative
final potential and does not permit any request to remain unserved.
\end{remark}

The free final configuration is essential: a prescribed terminal cache may
make proactive loading useful.

\section{Offline Structure and Algorithms}
\label{sec:offline}
\label{sec:offline-one-hole}
\label{sec:offline-general}

The causal lazy projection reduces exact optimization to schedules that load
only pending pages.  We first exploit the resulting one-dimensional state
when there is one hole, then show why classical victim normalization fails
and retain the full hole configuration for general $r$.

\subsection{Exact optimization with one hole}

Assume $m=k+1$, so the cache is represented by its unique hole $h$.  The
only additional state information is a bit $b\in\{0,1\}$ indicating whether
$h$ is pending.

\begin{theorem}[Exact one-hole dynamic program]
\label{thm:one-hole-dp}
The offline optimum can be computed with $2(k+1)$ states per arrival epoch,
$O(k)$ working space, and $O(nk)$ time.
\end{theorem}

\begin{proof}
Let $\{h_0\}=\Pages\setminus C_0$ define the unique initial hole.  Introduce a
pre-arrival layer at $t_0=t_1=0$ with
\[
  D_0(h_0,0)=0,
  \qquad
  D_0(h,b)=+\infty\quad\text{for }(h,b)\ne(h_0,0).
\]
For $i\ge1$, let $D_i(h,b)$ denote the minimum cost after processing the batch
at $t_i$ and all actions at that timestamp.  Before batch $A_i$, a state from
layer $i-1$ with pending bit $b$ incurs $b(t_i-t_{i-1})$.  After $A_i$
arrives, update
\[
  b\leftarrow b\lor [h\in A_i].
\]
The schedule may wait.  If $b=1$, it may instead load $h$, evict any
$v\ne h$, pay one, and enter state $(v,0)$.  Only states with $b=0$ are
accepted at the end.

For an $O(nk)$ implementation, keep clean values $U_h$ and pending
intercepts $P_h$, where a pending state's true cost at time $t$ is $P_h+t$.
At epoch $t_i$, for every state whose hole is requested, set
\[
  P_h\leftarrow\min\{P_h,U_h-t_i\},
  \qquad U_h\leftarrow+\infty.
\]
Using one common post-arrival snapshot, a service that creates hole $v$ gives
\[
  U_v\leftarrow
  \min\left\{U_v,\,1+t_i+\min_{h\ne v}P_h\right\}.
\]
The minimum and second minimum of the $P_h$ values evaluate all exclusions in
$O(k)$ time per layer.  In the intercept implementation, initialize
$U_{h_0}=0$ and every other $U_h$ and $P_h$ to infinity.  The answer is
$\min_h U_h$ after the final closure.
\end{proof}

\begin{remark}[Why the pending bit is necessary]
The hole alone is not a sufficient layered state.  With cache size one,
initial hole $A$, and requests $A@0$, $B@\varepsilon$, and $B@1.5$, the final
hole $A$ is reachable both as pending with a smaller accumulated label and as
clean with a larger accumulated label.  The pending path must still pay for a
terminal service and can be globally worse.  A page-only service-event DAG is
possible, but it stores the same distinction in long edges rather than
eliminating it.
\end{remark}

Lemma~\ref{lem:arrival-layer-normal-form} justifies the event layers above.
Moreover, Lemma~\ref{lem:lazy-optimum} supplies an optimum with no proactive
replacement; hence at most one replacement occurs at each epoch when there is
one hole.

\subsection{Physical farthest-next-use can fail}

Classical Belady does not supply a victim normal form for the delayed
problem.  For any action time $s$, write $\nu_s(x)$ for the next request to
$x$ strictly after $s$, evaluated after the complete batch when $s$ is an
arrival epoch, with value $\infty$ if none exists.  A physical
farthest-next-use rule evicts a cached page maximizing $\nu_s$, with a fixed
tie rule.

\begin{proposition}[Failure of physical farthest-next-use]
\label{ex:delayed-belady-fails}
Take $k=2$, page set $\{p,q,f\}$, and initial cache $\{q,f\}$.  The
singleton requests are
\[
  p@0,\qquad q@3,\qquad f@\tfrac72,
  \qquad p@5,\qquad q@10.
\]
At time zero, $\nu_0(q)=3<7/2=\nu_0(f)$, so a farthest-next-use rule uniquely
evicts $f$ when it loads $p$.  Nevertheless, the unrestricted optimum evicts
$q$ and has cost $5/2$, whereas every schedule using a farthest-next-use
victim at each replacement has cost at least $3$.
\end{proposition}

\begin{proof}
For the unrestricted upper bound, load $p$ and evict $q$ at time zero.  At
time $3$, page $q$ becomes pending.  The request to cached $f$ at time $7/2$
hits; immediately afterward, load $q$ and evict $f$.  The two movements cost
$2$, the only delay is $1/2$ for $q$, and the requests to $p@5$ and $q@10$
hit.

We next prove that $5/2$ is optimal.  By Lemma~\ref{lem:lazy-optimum}, it
suffices to consider nonproactive schedules.  Any schedule of cost below
$5/2$ has at most two replacements.  One replacement is infeasible: its
first action must load the initially missing page $p$, and cost below $5/2$
forces this action before time $3/2$, when either possible victim still has a
future request.  Thus there are exactly two replacements.  If the first
loads $p$ at time $s$, then $2+s<5/2$ gives $s<1/2$.  If it evicts $q$, the
second cannot reload $q$ before the complete
$f@7/2$ batch without evicting $p$ or $f$, both of which still have a future
request.  Hence the $q$-episode has delay at least $1/2$.  If the first
evicts $f$, the second cannot reload $f$ before the complete $p@5$ batch for
the same reason, so its delay is at least $3/2$.  In either case two
movements incur holding at least $1/2$, proving the lower bound $5/2$.

Now impose farthest-next-use on every replacement, including proactive ones.
Suppose for contradiction that a feasible restricted schedule costs below
$3$.  It has at most two replacements.  Zero is infeasible, and with one,
cost below $3$ forces the sole load of $p$ before time $2$, when either
possible victim has a later request.  Hence it has exactly two replacements.
Its first replacement
necessarily loads $p$ at some time $s<1$, and therefore uniquely evicts $f$
because $\nu_s(q)=3<7/2=\nu_s(f)$.  If the second replacement reloads $f$
proactively before time $7/2$, it evicts either $p$ or $q$, each of which has
a later request, so a third replacement is necessary.  Otherwise, let
$u\ge7/2$ denote the reload time.  Cost below $3$ implies
\[
  s+(u-7/2)<1,
\]
and hence $u<9/2<5$.  At that time
$\nu_u(p)=5<10=\nu_u(q)$, so farthest-next-use uniquely evicts $q$ and again
forces a third replacement at $q@10$.  This contradiction proves restricted
cost at least $3$.  Equality is attained by loading $p$ at $0$, $f$ at
$7/2$, and $q$ at $10$.

The failed exchange reverses its discrepancy at time $7/2$: replacing the
original $q$-load by an $f$-load improves the current delay but creates a
$q$-hole that must be served at time $10$.  A single banked movement cannot
both clear the $f$-episode and pay for this later reversal.
\end{proof}

Figure~\ref{fig:belady-failure} displays the two cache trajectories.  The
counterexample is temporal rather than a tie-breaking artifact: the victim
that returns sooner is the one the optimum deliberately evicts.

\begin{figure}[t]
  \centering
  \begin{tikzpicture}[
      x=0.9cm,
      y=1cm,
      >=Latex,
      cache/.style={draw,rounded corners=1.5pt,fill=blue!4,
                    inner sep=2.5pt,font=\scriptsize},
      action/.style={font=\scriptsize,align=center,fill=white,inner sep=1.2pt},
      request/.style={font=\scriptsize,anchor=north},
      track/.style={font=\scriptsize\bfseries,align=right,anchor=east,
                    text width=2.75cm}
    ]
    \draw[->] (0,0) -- (12.0,0) node[right] {$t$};
    \foreach \x/\lab in {0/0,3/3,3.5/{7/2},5/5,10/10} {
      \draw (\x,0.08)--(\x,-0.08);
      \node[request] at (\x,-0.12) {$\lab$};
    }
    \node[font=\scriptsize,anchor=south] at (0,0.1) {$p$};
    \node[font=\scriptsize,anchor=south] at (3,0.1) {$q$};
    \node[font=\scriptsize,anchor=south] at (3.5,0.1) {$f$};
    \node[font=\scriptsize,anchor=south] at (5,0.1) {$p$};
    \node[font=\scriptsize,anchor=south] at (10,0.1) {$q$};

    \node[track] at (-0.2,1.7) {unrestricted\\OPT};
    \node[cache] (u0) at (0.6,1.7) {$\{q,f\}$};
    \node[cache] (u1) at (4.0,1.7) {$\{p,f\}$};
    \draw[->] (u0)--node[action,above]{load $p$\\evict $q$} (u1);
    \node[cache] (u2) at (8.0,1.7) {$\{p,q\}$};
    \draw[->] (u1)--node[action,above]{at $7/2$: load $q$\\delay $1/2$} (u2);
    \node[action,anchor=west] at (8.85,1.7)
      {later requests hit; total $2+\tfrac12=\tfrac52$};

    \node[track] at (-0.2,3.35) {physical\\farthest-next-use};
    \node[cache] (b0) at (0.6,3.35) {$\{q,f\}$};
    \node[cache] (b1) at (4.0,3.35) {$\{p,q\}$};
    \draw[->] (b0)--node[action,above]{load $p$\\evict $f$} (b1);
    \node[cache] (b2) at (8.0,3.35) {$\{p,f\}$};
    \draw[->] (b1)--node[action,above]{at $7/2$: load $f$\\evict $q$} (b2);
    \node[cache] (b3) at (11.3,3.35) {$q$ reloaded};
    \draw[->] (b2)--node[action,above]{at $10$} (b3);
    \node[action,anchor=west] at (12.15,3.35) {total $3$};
  \end{tikzpicture}
  \caption{Why physical farthest-next-use fails in
  Proposition~\ref{ex:delayed-belady-fails}.  Evicting the sooner-requested
  page $q$ preserves the hit on $f$ and permits a half-unit delayed reload of
  $q$; the classical victim choice creates a third movement.  Cache boxes
  are spaced schematically; the arrow annotations give the action times.}
  \label{fig:belady-failure}
\end{figure}
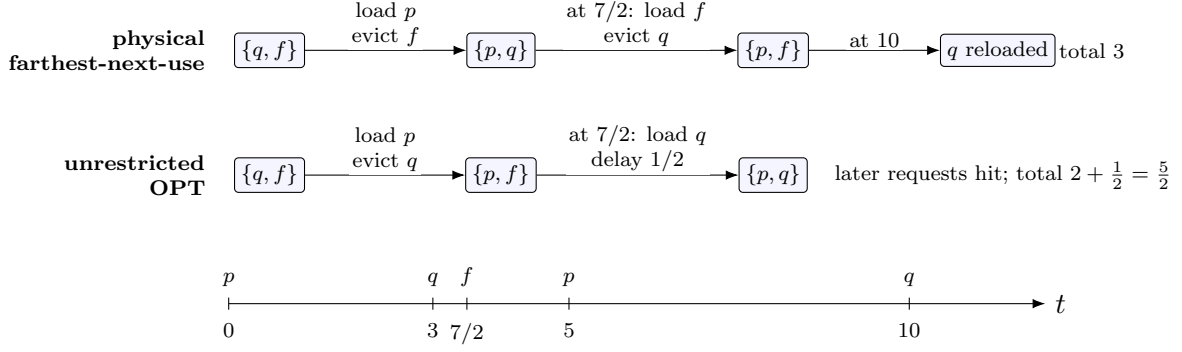

The proposition rules out a next-use-only rule for choosing physical victims.
Consequently, the exact offline state below must retain the current hole set:
future request times alone do not determine which cached page should be
evicted.

\subsection{Configuration dynamic program}

For general $r=m-k$, the exact state is
\[
  (H,B),\qquad |H|=r,\qquad B\subseteq H,
\]
where $H$ is the hole set and $B$ is its pending subset.  The Belady
counterexample above is why the state retains the entire set $H$: the victim
cannot be fixed from future request times alone.

At epoch $t_i$, the arrival update is
\[
  B\leftarrow B\cup(A_i\cap H).
\]
A lazy service selects $p\in B$ and a cached victim $v\notin H$, pays one,
and performs
\begin{equation}
  \label{eq:general-service-edge}
  H'=(H\setminus\{p\})\cup\{v\},
  \qquad B'=B\setminus\{p\}.
\end{equation}
Because $v$ was cached when the complete batch arrived, it is a clean hole
after eviction.  Every service edge decreases $|B|$ by one, so the zero-time
closure is a DAG and contains at most $r$ actions along a path.

Let $H_0=\Pages\setminus C_0$, put $t_0=t_1$, and define a virtual pre-first
layer by
\[
  D_0(H,B)=
  \begin{cases}
    0,&(H,B)=(H_0,\varnothing),\\
    +\infty,&\text{otherwise}.
  \end{cases}
\]
For $i=1,\ldots,n$, put $\Delta_i=t_i-t_{i-1}$.  Before the zero-time
service closure at epoch $i$, set
\[
  \widehat D_i(H,B')=
  \min_{\substack{B\subseteq H:\\
          B'=B\cup(A_i\cap H)}}
  \left\{D_{i-1}(H,B)+|B|\Delta_i\right\}.
\]
The minimum explicitly resolves collisions under the arrival map.  Since
$t_0=t_1$, the first layer receives no artificial pre-input holding charge.
Starting simultaneously from all finite labels $\widehat D_i(H,B)$, relax
all edges~\eqref{eq:general-service-edge} in decreasing order of $|B|$; the
resulting post-action labels are $D_i$.  After the final closure,
\[
  \OPT=\min_{|H|=r}D_n(H,\varnothing).
\]

\begin{theorem}[Exact general-hole dynamic program]
\label{thm:general-dp}
The recurrence is exact and has
\[
  S=\binom mr2^r
\]
states.  It can be implemented in
\[
  O\!\left(n\binom mr2^r(1+rk)\right)
\]
arithmetic relaxations and $O(S)$ numeric-label working space, in the
unit-cost indexed-state model where membership, a one-page swap, and successor
lookup take constant time.  With an explicit representation, the same
recurrence remains polynomial for fixed $r$, with the corresponding set and
indexing cost multiplying the edge bound.
\end{theorem}

\begin{proof}
Lemma~\ref{lem:lazy-optimum} restricts attention to lazy replacements.  The
displayed base layer is exactly the common initial configuration with no
pending episode.  Between arrivals, Lemma~\ref{lem:holding} gives exactly the
holding increment $|B|\Delta_i$.  At an arrival, the displayed arrival map is
forced, and taking the minimum preserves the best label when several states
coalesce.  Every
possible lazy replacement is exactly one edge of
\eqref{eq:general-service-edge}, and every such edge is feasible.  Since
$|B|$ strictly decreases, the closure is acyclic and its shortest paths
enumerate every possible zero-time service sequence exactly.  There are
$\binom mr$ choices of $H$ and $2^r$ choices of $B$.  Each state has at most
$rk$ service edges, yielding the stated bounds.
\end{proof}

\subsubsection{Same-epoch closure}

\begin{proposition}[Same-epoch closure]
\label{prop:same-epoch-closure}
Fix a post-arrival state $(H,B)$.  A nonempty lazy action word reaches
$(H',B')$ if and only if
\[
  H'\ne H,
  \qquad B'\subseteq B\cap H',
  \qquad H\setminus H'\subseteq B\setminus B'.
\]
Every such word uses exactly $|B|-|B'|$ replacements, and some word of that
length reaches every pair satisfying the conditions.
\end{proposition}

\begin{proof}
Every action loads and clears one current pending page and creates one clean
victim hole.  Thus an unserved pending page remains a hole, every initial
hole removed from $H$ was served, and the word length is $|B|-|B'|$.  A
nonempty word cannot return to $H'=H$: the initially cached victim of its
first action becomes a clean hole and, with no intervening arrival, cannot be
loaded later in the word.

Conversely, put
\[
  S=B\setminus B',\quad
  E=H'\setminus H,\quad
  R=H\setminus H',\quad
  I=S\cap H'.
\]
The conditions give $|E|=|R|\ge1$.  Match $E$ to $R$ and distribute the
pages of $I$ arbitrarily among the matched pairs.  For a resulting chain
$e,i_1,\ldots,i_ell,r$, first load $i_1$ while evicting $e$, then load
$i_2$ while evicting $i_1$, and so on, finally loading $r$ while evicting
$i_\ell$; for an empty chain, load $r$ while evicting $e$.  Every loaded page
belongs to $S$ and is still pending when loaded, while every victim is cached
at that micro-action.  The final holes are exactly $H'$, the remaining
pending set is $B'$, and every member of $S$ is loaded once.
\end{proof}

For fixed $r$, this is an XP algorithm with running time
$m^{O(r)}\operatorname{poly}(n)$; it is not an FPT result and does not rule out
a different polynomial or FPT algorithm.

\subsection{Exact support compression}

Let $R$ denote the set of pages requested at least once, let $d$ denote $|R|$,
and let $C_0$ denote the initial cache.  Define
\[
  \mu=|R\setminus C_0|.
\]

\begin{proposition}[Exact request-support quotient]
\label{prop:support-quotient}
The configuration dynamic program has an exact quotient whose states are
$(H_R,B)$, where $H_R\subseteq R$ is the set of requested holes and
$B\subseteq H_R$ is its pending subset, with
\[
  \max\{0,d-k\}\le |H_R|\le\mu.
\]
Its state count is
\[
  \sum_{h=\max\{0,d-k\}}^\mu \binom dh2^h,
\]
independent of the raw universe size.
\end{proposition}

\begin{proof}
Pages outside $C_0\cup R$ are silent initial holes.  A lazy schedule never
loads them, so they may be deleted.  The pages of $C_0\setminus R$ are never
requested and are interchangeable resident dummies.  If $h=|H_R|$, then the
cache contains $d-h$ requested pages and therefore $k-d+h$ dummy pages.  This
quantity is nonnegative, while the number of available dummies is
$|C_0\setminus R|=k-d+\mu$; equivalently,
$\max\{0,d-k\}\le h\le\mu$.

The arrival map acts only on $H_R$.  A service of $p\in B$ either evicts a
requested resident $q\in R\setminus H_R$, producing
$H_R'=H_R-p+q$, or, when $k-d+h>0$, evicts an arbitrary resident dummy and
produces $H_R'=H_R-p$.  In both cases $p$ is removed from $B$ and the new
hole is clean.  These are all quotient transitions.

Projecting a full lazy schedule to requested-page membership gives exactly
this trajectory.  Conversely, lift a requested-victim transition using its
named victim, and lift a dummy-victim transition using any currently resident
dummy.  Dummy counts show that such a choice exists, and dummy identity can
never affect a future request or service delay.  The lifted and quotient
schedules therefore have identical movement and holding costs.  Counting
$H_R$ and its pending subsets gives the displayed state count.
\end{proof}

If $k\ge d$, every initially missing requested page can be loaded once while
evicting a distinct dummy, and $\OPT=\mu$.  Thus the exponential state
growth is controlled by the pages that are actually requested, not by silent
pages in the full universe.

\section{Deterministic Online Algorithms for a General Page Universe}
\label{sec:online-general}

We now permit an arbitrary number of holes and hence many simultaneous
pending pages.  Fix $\theta>0$.  Each page starts a timer at the oldest
request of its current pending episode and is loaded exactly $\theta$ later.
Fix once and for all a total order $\prec$ on $\Pages$.  Recency keys are
compared lexicographically, with larger keys meaning more recent.  Every
initially cached page $p$ receives the artificial key
\[
  (-\infty,0,\operatorname{rank}_{\prec}(p)),
\]
and an actual arrival of $p$ at time $t$ assigns
$(t,0,\operatorname{rank}_{\prec}(p))$.  Thus simultaneous arrivals are tied
deterministically without imposing a physical order on the atomic batch.
Deadline ties and minimum-key victim ties use $\prec$.  We study two victim
rules.
\begin{description}[leftmargin=3.2cm,labelwidth=2.9cm,style=multiline]
\item[Arrival-LRU]
Loading a page does not update its key; its key remains its last actual
arrival.
\item[{\parbox[t]{2.9cm}{\raggedright Service-touch\\LRU}}]
The $j$th service micro-action at time $t$ assigns the loaded page key
$(t,1,j)$, after the complete arrival batch and in service micro-order.
\end{description}
Both rules permit a newly loaded page to be evicted by a later service at the
same timestamp.  Algorithm~\ref{alg:timer-lru} fixes the event order explicitly.

\begin{algorithm}[H]
\caption{Timer-LRU$(\theta)$, with either recency convention}
\label{alg:timer-lru}
\begin{algorithmic}[1]
\State $C\gets C_0$; initialize pending episodes and recency keys as above
\ForAll{event times $t$ in increasing order}
  \State Process the complete raw batch at $t$ against the pre-action cache
  $C$, and update arrival-recency keys
  \ForAll{pages $p$ whose request at $t$ starts a pending episode}
    \State Set the timer of $p$ to expire at $t+\theta$
  \EndFor
  \ForAll{timers of still-pending pages $p$ expiring at $t$, in fixed order}
    \State Choose a minimum-key page $q$ of $C$ under the chosen LRU rule
    \State Replace $q$ by $p$; clear the pending episode of $p$
    \If{using Service-touch LRU}
      \State Set the key of $p$ to the current service micro-order
    \EndIf
  \EndFor
\EndFor
\end{algorithmic}
\end{algorithm}

Every service has cost at most $1+\theta$.  It remains to bound the number
$N$ of services against an arbitrary feasible lazy comparator $O$.  Write
$M$ and $D$ for its movement and holding costs.  The \emph{busy interval} of
an online service activated at $a$ is the closed interval $[a,a+\theta]$;
touching an $O$ micro-action at either endpoint counts as crossing.  A service
for page $p$ is an activation miss iff $p\notin C_O^-(a)$, and an activation
hit otherwise.

\begin{lemma}[Activation misses and crossing hits]
\label{lem:det-activation-crossing}
The number of activation misses satisfies
\[
  N_{\rm miss}\le D/\theta+M.
\]
Partition the event order into half-open $O$-phases between consecutive
post-action cuts, on each of which its cache is fixed.  The number of
activation-hit timers that touch the next $O$ movement
is at most $kM$.
\end{lemma}

\begin{proof}
Consider an activation miss for page $p$ at time $a$.  Since
$p\notin C_O^-(a)$, the arrival at $a$ starts or joins a unique $O$-pending
episode $e$ for $p$.  By definition of the closed service span,
\[
  a\in\mathcal J_e=[a_e,s_e],
\]
including when $O$ serves the episode immediately at $a=s_e$.  Activations of
one online page are spaced strictly more than $\theta$ apart: its timer is
cleared at one deadline, and the arrival-first convention permits a new
activation only at a later arrival epoch.  Consequently, an $O$-episode of
length $L=s_e-a_e$ contains at most $L/\theta+1$ activations of its page.
This includes $L=0$, when $\mathcal J_e$ is a singleton.  Since $O$ is lazy
and feasible, its episodes are in bijection with its $M$ replacements.
Summing over them and using $\sum_e(s_e-a_e)=D$ proves the first bound.

For activation hits, charge every busy interval that touches the next $O$
movement to the first movement it touches.  At the activation the timer's
page lies in the fixed pre-movement cache, which has $k$ pages, and a page has
at most one live timer.  Thus at most $k$ timers are charged to each $O$
movement.
\end{proof}

The remaining timers are activation hits whose full busy intervals lie in a
single fixed-cache $O$ phase.  The complication is that the online cache may
temporarily contain a page outside the fixed $O$ cache even though $O$ has no
pending request for that page.  Figure~\ref{fig:quarantine} depicts the short
quarantine after which LRU removes precisely these clean exceptions.

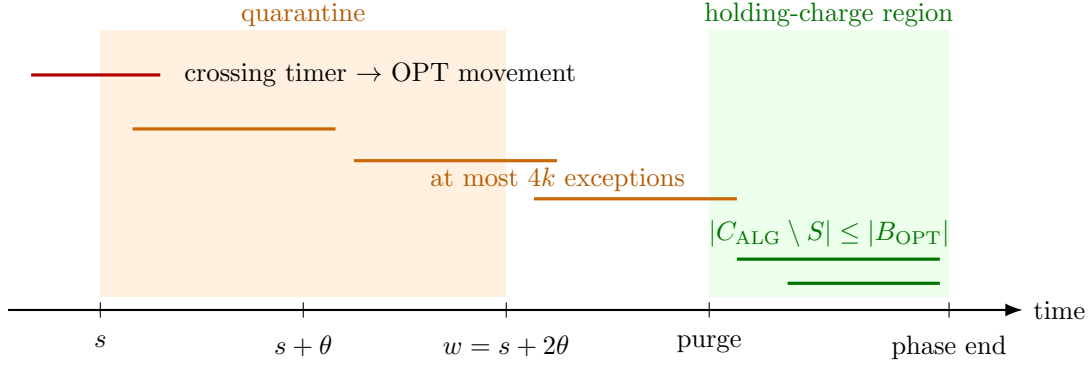
\begin{figure}[t]
  \centering
  \begin{tikzpicture}[x=1.22cm,y=0.84cm,>=Latex,font=\small]
    \draw[->,thick] (0,0) -- (11.0,0) node[right] {time};
    \draw (1,0.12) -- (1,-0.12) node[below=3pt] {$s$};
    \draw (3.2,0.12) -- (3.2,-0.12) node[below=3pt] {$s+\theta$};
    \draw (5.4,0.12) -- (5.4,-0.12) node[below=3pt] {$w=s+2\theta$};
    \draw (7.6,0.12) -- (7.6,-0.12) node[below=3pt] {purge};
    \draw (10.2,0.12) -- (10.2,-0.12) node[below=3pt] {phase end};

    \draw[fill=orange!12,draw=none] (1,0.2) rectangle (5.4,4.4);
    \node[orange!70!black] at (3.2,4.65) {quarantine};
    \draw[fill=green!8,draw=none] (7.6,0.2) rectangle (10.2,4.4);
    \node[green!45!black] at (8.9,4.65) {holding-charge region};

    \draw[red!70!black,very thick] (0.25,3.70) -- (1.65,3.70);
    \node[anchor=west] at (1.8,3.70) {crossing timer $\to$ OPT movement};

    \draw[orange!80!black,very thick] (1.35,2.85) -- (3.55,2.85);
    \draw[orange!80!black,very thick] (3.75,2.35) -- (5.95,2.35);
    \draw[orange!80!black,very thick] (5.7,1.75) -- (7.9,1.75);
    \node[anchor=east,orange!70!black] at (7.45,2.05)
      {at most $4k$ exceptions};

    \draw[green!45!black,very thick] (7.9,0.80) -- (10.1,0.80);
    \draw[green!45!black,very thick] (8.45,0.42) -- (10.1,0.42);
    \node[green!45!black] at (8.9,1.20)
      {$|C_{\ALG}\setminus S|\le|B_{\OPT}|$};
  \end{tikzpicture}
  \caption{An OPT phase with fixed cache $S$.  Timers touching the phase
  boundary are charged to the crossing OPT movement.  A length-$2\theta$
  quarantine and at most $k$ subsequent services purge old clean exceptions;
  timers straddling the purge are also exceptional.  Afterwards every
  outsider in the online cache is OPT-pending, so the remaining timer length
  is charged by integration to OPT holding.}
  \label{fig:quarantine}
\end{figure}

\paragraph{Proof outline.}
Activation misses pay to offline holding or movement, and hit timers that
reach an offline cache change pay to that change.  Inside an interval where
the offline cache is fixed, the only obstacle is an online page outside that
cache that is not offline-pending.  After waiting $2\theta$, LRU makes every
such old exception eligible for removal, and at most $k$ later services remove
all of them.  At most $4k$ timers are exceptional: at most $2k$ start during
the waiting period, at most $k$ finish between its end and the purge, and at
most $k$ remain active across the purge.  Afterward, online pages outside the
offline cache are offline-pending, so regular timer length can be charged to
offline holding.

\begin{lemma}[Quarantine and purge]
\label{lem:det-quarantine}
In every noninitial fixed-cache $O$ phase, at most $4k$ contained
activation-hit timers are exceptional.  After these exceptional timers and
the ensuing purge, every page of $C_{\ALG}\setminus S$ is $O$-pending, where
$S$ is the fixed $O$ cache of the phase.  The initial phase has no clean
exception.
\end{lemma}

\begin{proof}
Let $s$ denote the real time of the phase start, let $\xi_s$ be the
post-action event-order cut that starts the phase, set $w=s+2\theta$, and call a page in
$C_{\ALG}\setminus S$ a \emph{clean exception} if it is not $O$-pending.
We first record when such an exception can be inserted.  Suppose that an
online service during the phase inserts $p\notin S$ and that $p$ is not
$O$-pending at the insertion time.  If its timer activation occurred after
$\xi_s$, then the activating request would miss the fixed cache $S$ and start
an $O$-pending episode.  That episode could not end before the insertion:
loading $p$ would be an $O$ replacement and hence would end the phase.  This
is a contradiction.  Thus the activation precedes $\xi_s$; in real time its
timestamp $a$ satisfies $a\le s$, and the insertion occurs by $s+\theta$.

A clean exception also receives no arrival after the phase-start cut
$\xi_s$: such an arrival misses $S$ and makes the page $O$-pending for the
rest of the phase.
Consequently, under Arrival-LRU its key is at most $s$, and under
Service-touch LRU every insertion that could give it a newer key occurs by
$s+\theta$.  Hence every clean exception present after $s+\theta$ has LRU
key at most $s+\theta$, and no new clean exception can be created later in
the phase.

Every service strictly after $w$ was activated after $s+\theta$ and inserts
a page with a strictly newer key.  Thus after at most $k$ such services all
old clean exceptions have been evicted.  If the phase contains fewer than
$k$ further services, declare all of them exceptional and there is no
post-purge region to analyze.  Otherwise call the first action after which no
old clean exception remains the purge action.  At most $2k$ contained hit
timers activate during the length-$2\theta$ quarantine: their pages belong to
$S$, and successive activations of one page are separated by more than
$\theta$.  The purge occurs within the first $k$ services after $w$, so at
most $k$ hit timers finish after $w$ but no later than the purge action.
Finally, at most one live hit timer belongs to each page of $S$, so at most
$k$ remain active across that action.  Hence at most $4k$ contained timers
are exceptional.

The purge removes every clean exception present at that point.  No later
service can create a new one, and an outsider that is $O$-pending cannot
become clean before the next $O$ replacement.  Hence
$C_{\ALG}\setminus S\subseteq B_O$ throughout the post-purge region.  In
the initial phase the caches agree, and every subsequently inserted outsider
starts an $O$-pending episode that cannot end while the initial $O$ cache
remains fixed; thus the same invariant holds from the start.
\end{proof}

\begin{lemma}[Post-purge holding charge]
\label{lem:det-post-purge}
Apart from the exceptional timers of Lemma~\ref{lem:det-quarantine}, the
number of contained activation-hit services satisfies
\[
  N_{\rm hit}^{\rm contained,regular}\le D/\theta.
\]
\end{lemma}

\begin{proof}
After the purge, an active contained hit timer requests a page
$p\in S\setminus C_{\ALG}$: it lies in $S$ at activation and remains absent
from the online cache until its deadline.  Equal cache sizes give
\[
  \#\{\text{active regular timers}\}
  \le |S\setminus C_{\ALG}|
  =|C_{\ALG}\setminus S|
  \le |B_O|,
\]
where the last inequality is the post-purge invariant and $B_O$ is the
set of $O$-pending pages.  Integrate this pointwise inequality over all $O$
phases.  Every regular timer contributes its full length $\theta$ to the
left-hand side, while Lemma~\ref{lem:holding} identifies the integral on the
right with $D$.
\end{proof}

\begin{theorem}[Unified deterministic bound]
\label{thm:general-lru}
For either recency rule and every feasible lazy offline comparator $O$ with
movement $M$ and holding cost $D$,
\begin{equation}
  \label{eq:service-count-general}
  N\le \frac{2D}{\theta}+(5k+1)M.
\end{equation}
Consequently, choosing the lazy optimum guaranteed by
Lemma~\ref{lem:lazy-optimum},
\[
  \ALG\le
  (1+\theta)\max\left\{\frac2\theta,5k+1\right\}\OPT.
\]
In particular, the choice $\theta=2/(5k+1)$ gives
$\ALG\le(5k+3)\OPT$.
\end{theorem}

\begin{proof}
Lemmas~\ref{lem:det-activation-crossing}--\ref{lem:det-post-purge} and the
fact that each $O$ movement starts at most one noninitial phase give
\[
  N_{\rm hit}
  \le D/\theta+4kM+kM
  =D/\theta+5kM.
\]
Adding the activation-miss bound proves~\eqref{eq:service-count-general}.
Each of the $N$ services pays one movement and holding at most $\theta$.
Multiplying the service-count bound by $1+\theta$ and balancing its two
coefficients proves the theorem.  A service exactly at $w$ belongs to the
quarantine, and simultaneous services are the ordered zero-time microstates
specified in Algorithm~\ref{alg:timer-lru}.
\end{proof}

\subsection{Tightness of the deterministic order}

The two variants can behave differently.  With $k=r=2$, initial holes
$\{0,1\}$, batches $\{0,1,2\}@0$ and $\{0\}@2$, threshold one, and a suitable
fixed tie order, Arrival-LRU pays six while Service-touch LRU pays four and
OPT pays two.  Arrival-LRU immediately re-evicts page $0$ because its actual
arrival key is old; Service-touch protects it.

Neither rule improves the asymptotic lower bound.  Start with cached pages
$c_1,\ldots,c_k$ and a hole $x$.  Request
\[
  x,c_1,c_2,\ldots,c_{k-1}
\]
with every request arriving after the preceding deadline.  LRU performs $k$
services, while OPT immediately loads $x$, evicts $c_k$, and hits every later
request.  The ratio is $k(1+\theta)=\Omega(k)$.

\begin{proposition}[Deterministic lower-bound order]
Every deterministic online algorithm has competitive ratio $\Omega(k)$.
\end{proposition}

\begin{proof}
Run the standard adaptive classical paging adversary on $k+1$ pages
\cite{sleator1985amortized}.  After the algorithm has cleared the current
pending episode, issue the next request, at a strictly later time, to a page
outside its current cache.  No future request is revealed while the current
one is pending, so waiting supplies no information and only adds holding
cost.  Every adversarial request therefore forces at least the movement paid
by the induced classical execution.  The delayed offline optimum can serve
the resulting classical sequence immediately and pay its classical paging
cost.  The usual phase lower bound, and hence its $\Omega(k)$ ratio, transfers
unchanged.
\end{proof}

\section{Temporal Aggregation for a General Page Universe}
\label{sec:online-general-randomized}

The aggregation below has both a deterministic offline use and a randomized
online use.  The main online difficulty with a literal
delayed Marker is that its service stream
depends on earlier random evictions.  Even for an oblivious raw input, that
derived stream can therefore behave adaptively.  We avoid this correlation by
performing temporal aggregation \emph{before} running a classical paging
algorithm.
This separates the timing step from the later eviction choices: the
aggregation map is cache-independent, while the later paging algorithm
sees only a virtual sequence fixed by the input.  The claim is not
that timing and eviction decouple in every physical schedule; the
deterministic algorithm of Section~\ref{sec:online-general} handles their
coupling directly.

Figure~\ref{fig:temporal-reduction} summarizes the four layers.  The first
arrow is cache-independent.  The second may use either an exact deterministic
paging optimum or a randomized online paging algorithm.  The last arrow is
pathwise and never performs a proactive physical replacement.

\begin{figure}[t]
  \centering
  \begin{tikzpicture}[
      >=Latex,
      font=\small,
      box/.style={draw,rounded corners=2pt,align=center,minimum height=2.0cm,
                  text width=2.7cm,fill=blue!5},
      arrow/.style={->,thick,blue!55!black}
    ]
    \node[box] (raw) {
      \textbf{Raw timed input}\\[2pt]
      input-only windows\\
      $[a_p,a_p+\theta]$
    };
    \node[box,right=0.72cm of raw] (virtual) {
      \textbf{Virtual sequence}\\[2pt]
      one representative\\
      $p@(a_p+\theta)$
    };
    \node[box,right=0.72cm of virtual] (shadow) {
      \textbf{Shadow paging}\\[2pt]
      lazy algorithm $A$\\
      cache $S$
    };
    \node[box,right=0.72cm of shadow] (physical) {
      \textbf{Physical cache}\\[2pt]
      if $p$ is pending,\\
      swap $q\in C\setminus S$ to $p$\\[-1pt]
      {\scriptsize $g+\Delta\Phi\le f$}
    };

    \draw[arrow] (raw) -- (virtual);
    \draw[arrow] (virtual) -- (shadow);
    \draw[arrow] (shadow) -- (physical);

    \node[below=0.34cm of shadow,align=center,font=\scriptsize,xshift=1.65cm]
      {A shadow fault with no physical pending episode changes only $S$.};
  \end{tikzpicture}
  \caption{The common temporal reduction.  The first two layers produce an
  ordinary-paging instance fixed independently of all cache and random
  choices.  The last layer is a pathwise,
  pending-only projection.  Choosing exact classical paging gives the
  offline $5$-approximation; at $\theta=2/3$, choosing an
  $H_k$-competitive classical algorithm gives the randomized $5H_k$ bound.}
  \label{fig:temporal-reduction}
\end{figure}

\subsection{Oblivious aggregation windows}

Fix $\theta>0$.  Independently for each page $p$, the first raw arrival not
already covered by an open window anchors a window
\[
  [a_p,d_p]=[a_p,a_p+\theta].
\]
All further arrivals to $p$ at times at most $d_p$, including an arrival at
the right endpoint, join this window.  After processing the complete raw
batch at $d_p$, output one \emph{virtual request} $p@d_p$ and close the
window.  A new $p$-window can start only at a time strictly larger than
$d_p$.  Simultaneous virtual requests use a fixed page order.

The resulting virtual sequence $\sigma_\theta$ is determined entirely by the
raw input and $\theta$; it is independent of every cache state and random
choice.

\paragraph{Algorithm AW-$A(\theta)$.}
Run a lazy classical paging algorithm $A$ on $\sigma_\theta$ in a
\emph{shadow cache} $S$, initialized to the common initial cache.  Maintain a
possibly different physical cache $C$, with $C=S$ initially.  Raw cache hits
are served immediately and do not touch $S$.  At a virtual request $p@d_p$,
first let $A$ process $p$ in $S$.  If $p$ is currently pending in the
physical system, then $p\notin C$ and $p\in S$; choose, by a fixed tie rule,
some $q\in C\setminus S$ and physically replace $q$ by $p$.  If $p$ is not
pending, take no physical action, even if the shadow request faults.

\begin{algorithm}[H]
\caption{Cache-independent aggregation with pending-only projection
  (AW-$A(\theta)$)}
\label{alg:aw}
\begin{algorithmic}[1]
\State $S\gets C_0$; $C\gets C_0$; initialize no open page windows
\ForAll{raw arrival epochs and window deadlines $t$ in increasing order}
  \State Process the complete raw batch at $t$ against $C$; hits are served
  immediately and misses start or join page-specific pending episodes
  \ForAll{raw arrivals $p@t$ in the batch}
    \If{$p@t$ is not covered by an open $p$-window}
      \State Open the window $[t,t+\theta]$ for $p$
    \EndIf
  \EndFor
  \State Let $R_t$ denote the representatives whose windows end at $t$
  \ForAll{$p@t\in R_t$ in the fixed page order}
    \State Let $A$ process $p$ in the shadow cache $S$
    \If{$p$ is physically pending}
      \State Choose $q\in C\setminus S$ by the fixed tie rule and set
      $C\gets C-q+p$
      \State Clear the pending episode of $p$
    \EndIf
    \State Close the $p$-window
  \EndFor
\EndFor
\end{algorithmic}
\end{algorithm}

\begin{lemma}[Nonproactive shadow projection]
\label{lem:aw-physical}
If $F$ is the number of shadow faults and $G$ is the number of physical
replacements, then pathwise
\[
  G\le F
  \qquad\text{and}\qquad
  \ALG\le (1+\theta)G\le (1+\theta)F.
\]
Every physical replacement is nonproactive.
\end{lemma}

\begin{proof}
Put $\Phi=|S\setminus C|$.  A shadow hit leaves $S$ unchanged, and a shadow
fault changes one page of $S$, so the shadow update alone increases $\Phi$
by at most one.  If $p$ is physically pending after that update, then
$p\in S\setminus C$.  Since $S$ and $C$ are equal-size $k$-sets, a page
$q\in C\setminus S$ exists.  The physical swap $q\to p$ costs one and
decreases $\Phi$ by exactly one.  If $f_j,g_j\in\{0,1\}$ indicate the shadow
fault and physical action at the $j$th representative, respectively, then
\[
  g_j+\Phi_j-\Phi_{j-1}\le f_j.
\]
The caches agree initially.  Summing in the fixed order of all
representatives, including simultaneous ones, and using
$\Phi_{\rm final}\ge0$ gives $G\le F$.

It remains to check service.  If a raw request in a $p$-window finds $p$
absent, it starts or joins a pending episode.  Before $p@d_p$, the algorithm
cannot load $p$: physical actions load only the page of the current
representative, and the window contains no earlier $p$-representative.
Hence the episode is still pending at $d_p$ and is cleared by the prescribed
physical swap.  Inductively, the preceding $p$-representative cleared every
earlier physical $p$-episode, and the next anchor is strictly later than that
deadline.  Thus the current episode's oldest request lies in $[a_p,d_p]$, so
its delay is at most $\theta$.  Every physical swap is triggered by such an
episode and is therefore nonproactive.  Thus each of the $G$ physical
replacements costs at most $1+\theta$.
\end{proof}

\begin{remark}[Endpoint and terminal conventions]
\label{rem:aggregation-boundaries}
At a time shared by raw arrivals and window deadlines, the complete batch is
processed first and representatives are then processed in their fixed order.
An earlier representative may evict a page whose request in that batch was a
hit, but cannot make that request retroactively pending.  Each page has at
most one representative at a timestamp because its next window can start
only strictly after the current deadline.  After the last raw arrival, all
remaining deadlines are still processed; the online algorithm needs no
end-of-input notification.  These conventions also cover empty deadline
batches and zero-age service at a shared endpoint.
\end{remark}

\subsection{Comparison with a delayed optimum}

Fix any feasible delayed schedule $O$.  Let $M$ and $D$ denote its movement and
holding costs, and let $O^+(t)$ denote its cache after the raw batch and all
of its ordered actions at time $t$.  Call a representative $p@d_p$
\emph{bad} when $p\notin O^+(d_p)$, and let $Z$ denote the number of bad
representatives.

A bad window has one of two certificates: either its page stays absent for
the whole window, yielding $\theta$ units of holding, or the page appears
during the window and has a last eviction before the post-action deadline
state.  The next lemma charges these two cases separately.

\begin{lemma}[Bad-representative charge]
\label{lem:bad-representatives}
\[
  Z\le \frac{D}{\theta}+M.
\]
\end{lemma}

\begin{proof}
Consider a bad window $[a_p,d_p]$.  If $p$ is absent when the batch at $a_p$
arrives and belongs to no intermediate cache microstate through the
post-action state at $d_p$, then the
request at $a_p$ remains pending for the full $\theta$ time units.  Charge the
representative to that segment of page-specific holding time.

Otherwise $p$ is present initially or is loaded somewhere in the window.
Since it is absent in $O^+(d_p)$, the window contains a last transition from
present to absent; charge the representative to that eviction of $p$.  This
also covers zero-time load--evict sequences at an endpoint.

Successive windows of one page have disjoint interiors and the next anchor is
strictly later than the preceding deadline.  Hence the page-labelled holding
segments $[a_p,d_p)\times\{p\}$ are disjoint for each page, while segments of
different pages are correctly added by Lemma~\ref{lem:holding}.  Each replacement evicts one page
and can be the last eviction for at most one window of that page.  The two
classes contribute at most $D/\theta$ and $M$, respectively.
\end{proof}

\begin{lemma}[Virtual paging comparator]
\label{lem:virtual-comparator}
If $\OPT_{\mathrm{pg}}(\sigma_\theta)$ is the classical paging optimum for
the virtual sequence with the common initial cache, then
\[
  \OPT_{\mathrm{pg}}(\sigma_\theta)
  \le M+2Z
  \le \frac{2D}{\theta}+3M.
\]
\end{lemma}

\begin{proof}
Construct a proactive unit-cost paging schedule.  Between virtual requests it
mimics every replacement of $O$.  At time $d$, first mimic all actions of $O$
at that timestamp and then process simultaneous virtual requests in their
fixed order.  A good representative hits.  For a bad $p@d$, load $p$ while
evicting an arbitrary cached page $q$, serve $p$, and immediately reload $q$
while evicting $p$.  This two-move detour restores $O^+(d)$, so simultaneous
detours concatenate and the total cost is at most $M+2Z$.

For completeness, couple this proactive cache $A$ to a lazy cache $L$ and
put $\Phi=|A\setminus L|$.  A proactive unit swap changes $\Phi$ by at most
one.  At a virtual request served while $p\in A\setminus L$, equal cache
sizes provide a page $q\in L\setminus A$; the lazy swap $q\to p$ costs one
and decreases $\Phi$ by one.  If $p\in L$, the lazy schedule does nothing.
Summing over the ordered micro-actions and request-service instants gives
\[
  \operatorname{cost}(L)+\Phi_{\rm final}
  \le \operatorname{cost}(A)+\Phi_{\rm initial}.
\]
The caches agree initially and the terminal cache is free, so the ordinary
classical optimum is no larger than the proactive comparator.  Lemma
\ref{lem:bad-representatives} gives the second inequality.
The post-action cut in the definition of a bad representative is essential:
it includes a same-time load followed by eviction, which is charged to the
last present-to-absent transition in Lemma~\ref{lem:bad-representatives}.
\end{proof}

\subsection{Deterministic offline use}

\begin{theorem}[Polynomial-time offline approximation]
\label{thm:offline-five-approx}
For every finite page universe and every cache size $k$, offline paging with
per-replacement maximum delay admits a deterministic nonproactive
polynomial-time $5$-approximation.
\end{theorem}

\begin{proof}
Construct $\sigma_{2/3}$ and compute a deterministic lazy classical paging
optimum on it, for example by Belady's farthest-in-future rule with fixed tie
breaking.  Denote its number of faults by
$F^*=\OPT_{\rm pg}(\sigma_{2/3})$, run it as the shadow execution of
AW-$A(2/3)$, and use the nonproactive projection of
Lemma~\ref{lem:aw-physical}.  This gives
\[
  \ALG_{\rm off}\le\frac53F^*.
\]
Apply Lemma~\ref{lem:virtual-comparator} to an optimal delayed schedule with
movement $M^*$ and delay $D^*$.  Since $2/(2/3)=3$,
\[
  F^*\le3D^*+3M^*=3\OPT.
\]
Consequently $\ALG_{\rm off}\le5\OPT$.

There are at most as many representatives as raw page occurrences.
Aggregation and event merging take polynomial time, and Belady's algorithm
can be implemented in $O(N\log(k+1))$ time after next occurrences are
precomputed, where $N=|\sigma_{2/3}|$.  Exact rational comparisons, or an
equivalent exact time-comparison oracle, handle arrivals at window endpoints.
\end{proof}

The approximation is nonproactive and does not round the marginal LP of
Section~\ref{sec:approximation}.  Its shadow cache is only an internal
computational state.

\begin{remark}[Limit of the present reduction]
The factor five is the best value obtainable by optimizing $\theta$ in the
proved coefficient bound
\[
  (1+\theta)
  \left(\frac{2D}{\theta}+3M\right).
\]
Indeed, its delay coefficient decreases and its movement coefficient
increases until $2/\theta=3$, namely $\theta=2/3$.  Any improvement through
this framework must therefore strengthen physical feasibility, the
bad-representative charge, or the two-move virtual comparator; retuning the
window length alone cannot improve five.
\end{remark}

\subsection{Randomized online use}

\begin{theorem}[Oblivious aggregation reduction]
\label{thm:aggregation-reduction}
Suppose a lazy classical paging algorithm $A$ is $\alpha$-competitive against
the adversary that fixes $\sigma_\theta$.  For every feasible delayed
comparator $O$ with movement $M$ and holding cost $D$, running the
nonproactive shadow projection after the above aggregation gives
\[
  \E[\ALG]
  \le
  \alpha(1+\theta)
  \left(\frac{2D}{\theta}+3M\right)
\]
against an oblivious raw-input adversary.
\end{theorem}

\begin{proof}
The virtual sequence is a deterministic function of the raw input, so the
classical guarantee applies without conditioning on the algorithm's random
cache.  Combine Lemmas~\ref{lem:aw-physical} and
\ref{lem:virtual-comparator}.  The argument holds for every feasible $O$;
minimizing the right-hand side gives the corresponding optimal-comparator
bound.
\end{proof}

McGeoch and Sleator's partitioning algorithm is a lazy classical paging
algorithm and, for every fixed request sequence $\sigma$ with a common
initial cache, satisfies the strict guarantee
\begin{equation}
  \label{eq:partition-classical}
  \E[F_{\rm Partition}(\sigma)]
  \le H_k\OPT_{\rm pg}(\sigma).
\end{equation}
The expectation is only over the algorithm's random choices; there is no
additive term \cite{mcgeoch1991strongly}.

\begin{corollary}[General-universe randomized bound]
\label{cor:general-partition}
For every finite page universe, every $k\ge1$, and every number of holes,
AW-Partition$(2/3)$ satisfies
\[
  \E[\ALG]\le 5H_k\OPT
\]
against an oblivious adversary.
\end{corollary}

\begin{proof}
The raw input fixes $\sigma_{2/3}$ before the partitioning algorithm's random
choices, so~\eqref{eq:partition-classical} supplies $\alpha=H_k$ in
Theorem~\ref{thm:aggregation-reduction}.  At $\theta=2/3$,
\[
  (1+\theta)\frac{2}{\theta}=5
  \qquad\text{and}\qquad
  3(1+\theta)=5,
\]
which gives the claim without an additive constant.
\end{proof}

\paragraph{A simpler Marker version.}
For Marker, the phases are the maximal consecutive blocks of
$\sigma_\theta$ containing at most $k$ distinct pages.  These, rather than
phases of the uncompressed raw arrivals, are the phase and clean-page notions
used below.  If $P_i$ is the distinct-page set of phase $i$, define
$c_1=|P_1\setminus C_0|$ and
$c_i=|P_i\setminus P_{i-1}|$ for $i\ge2$, and put $H_0=0$.

For completeness, the classical clean-page lemma gives Marker
\[
  \E[F]\le 2H_k\OPT_{\mathrm{pg}}(\sigma_\theta).
\]
Indeed, if $c_i$ pages are clean in virtual phase $i$, deferred decisions on
the unmarked old pages give expected phase cost at most
\[
  c_i+c_i(H_k-H_{c_i})\le c_iH_k.
\]
For any paging schedule with phase costs $f_i$, the union of two consecutive
full phase sets has size $k+c_i$, so $c_1\le f_1$ and
$c_i\le f_{i-1}+f_i$ for $i\ge2$.  Therefore
$\sum_i c_i\le2\sum_i f_i$.  The common initial cache removes the usual
first-phase additive term.

\begin{corollary}[Explicit Marker bound]
\label{cor:general-marker}
For every finite page universe, every $k\ge1$, and every number of holes,
AW-Marker$(2/3)$ satisfies $\E[\ALG]\le 10H_k\OPT$ against an oblivious
adversary.
\end{corollary}

\begin{proof}
Use $\alpha=2H_k$ in Theorem~\ref{thm:aggregation-reduction}.  At
$\theta=2/3$, the coefficients of $D$ and $M$ both become $10H_k$.
\end{proof}

\begin{remark}[Why the shadow cache matters]
If a page was hit early in its aggregation window but evicted before the
representative, the representative can fault in the shadow cache while no
physical request is pending.  The algorithm still processes that shadow
fault but performs no physical replacement.  Cancelling the representative
in the shadow execution would make the processed virtual stream depend on
earlier random evictions and would bring back the adaptive-sequence
problem.  The projection lemma avoids proactive physical actions without
making this cancellation.
\end{remark}

\begin{remark}[Why direct multi-hole service phases fail]
With $k\ge3$ and at least $k+1$ initially missing pages requested in one
batch, their common deadline produces a completed reset block containing
$k$ deterministic services, already larger than $H_k$.  Reset windows may
also have unbounded overlap when many holes share a deadline.  Virtual phases
avoid both failures: their harmonic lemma scales with the number of clean
representatives, and their sequence is fixed before the victim coins.
\end{remark}

\subsection{Randomized lower-bound order}

\begin{proposition}
Already on a universe of $k+1$ pages, every randomized online algorithm has
competitive ratio $\Omega(H_k)$ against an oblivious adversary.
\end{proposition}

\begin{proof}
Use a finite-support classical paging lower-bound distribution in which every
sequence has length at most $L$ and every deterministic classical algorithm
has expected cost at least $cH_k$ times the expected classical optimum, for an
absolute constant $c>0$ \cite{fiat1991competitive}.  Place its consecutive
requests at times $0,T,2T,\ldots$ with $T=L+1$.

Fix a deterministic delayed algorithm $B$, and let $B'$ denote its causal lazy
projection from Lemma~\ref{lem:lazy-optimum}.  Construct a deterministic
classical algorithm $A_{B'}$ as follows.  At a classical fault at epoch $jT$,
process the request at $jT$ and then simulate $B'$ through the empty time
until $(j+1)T$.  If $B'$ clears the miss before that endpoint, make its
victim choice immediately in the classical cache.  Otherwise make any fixed
lazy victim choice and permanently switch to a fixed lazy paging rule.  This
simulation uses the known gap $T$ but not the identity of the next page; a
service exactly at $(j+1)T$ is crossing under the arrival-first convention.

Before the first crossing, $A_{B'}$ and $B'$ make the same movements.  On a
crossing path, the delayed episode alone contributes $T>L$
holding, whereas from that request onward the classical fallback incurs at
most $L$ faults.  Therefore, pathwise,
\[
  \operatorname{cost}(A_{B'})
  \le \operatorname{cost}(B')
  \le \operatorname{cost}(B).
\]
Taking expectations and applying the finite-support
classical lower bound gives
\[
  \E[\ALG_{\rm delay}]
  \ge \E[\ALG_{\rm pg}]
  \ge cH_k\,\E[\OPT_{\rm pg}].
\]
Serving every classical fault immediately is feasible in the delayed model,
so $\OPT_{\rm delay}\le\OPT_{\rm pg}$ for each sequence.  Yao's principle now
gives the claimed $\Omega(H_k)$ lower bound against an oblivious adversary.
\end{proof}

\section{Weighted Fetch Costs}
\label{sec:weighted}

We now allow pages to have different fetch costs.  Each page $p$ has a
weight $w_p>0$, known to the algorithm.  Loading $p$, whether to clear a
pending episode or proactively, costs $w_p$; an episode served at time $s$
still contributes the unweighted age $s-a_p$.  Thus a lazy service of $p$
costs
\[
  w_p+(s-a_p).
\]
As in the unit-cost model, time is measured in cost-normalized units: one
unit of page-episode age contributes one unit to the objective, and $w_p$ is
expressed on the same scale.
Let $M_w$ denote the sum of fetch costs, keep $D$ for the episode-delay
cost, and let $\OPT_w$ denote the weighted delayed optimum.  Define
\[
  w_{\min}=\min_{p\in\Pages}w_p,
  \qquad
  w_{\max}=\max_{p\in\Pages}w_p,
  \qquad
  \rho=\frac{w_{\max}}{w_{\min}}.
\]
The unweighted model is the case $w_p=1$ for every page.

This is the standard fetch-cost convention for weighted paging
\cite{young1994kserver,bansal2007weighted}.  Charging the evicted page
instead changes the cost of a trajectory by a boundary term involving the
initial and final caches.  To avoid importing such a boundary term into the
stated guarantees, we keep the fetch-cost convention throughout this section.

The main message is two-sided.  The structural results and the exact
fixed-number-of-holes algorithms survive without a loss.  For arbitrary page
sets, the
temporal reduction also survives, but its guarantees depend on the weight
spread $\rho$.  A strict spread-dependent lower bound appears already when
the cache has one slot.

\subsection{Weighted laziness and exact fixed-number-of-holes algorithms}

The one-sided discrepancy used in Section~\ref{sec:model} has a direct
weighted form.

\begin{lemma}[Weighted causal lazy projection]
\label{lem:weighted-lazy}
Every feasible weighted schedule $O$ has a causal nonproactive projection
$L$ such that
\[
  M_w(L)\le M_w(O),
  \qquad
  D(L)\le D(O).
\]
Consequently, proactive loading cannot improve the weighted offline optimum
or a competitive guarantee against an oblivious input.
\end{lemma}

\begin{proof}
Use $O$ as a shadow execution with cache $S$, maintain a physical cache $C$,
and set
\[
  \Phi_w(S,C)=\sum_{p\in S\setminus C}w_p.
\]
After a raw batch, if a physically pending page $p$ belongs to $S\setminus
C$, choose any $q\in C\setminus S$ and physically replace $q$ by $p$.
This repair costs $w_p$ and decreases $\Phi_w$ by exactly $w_p$, so its cost
and potential change sum to zero.

Next consider one shadow replacement that loads $p$.  It costs $w_p$ and,
before any physical repair, can increase $\Phi_w$ by at most $w_p$.  If $p$
is physically pending after the shadow update, then $p\in S\setminus C$;
the same repair costs $w_p$ and decreases $\Phi_w$ by exactly $w_p$.
Therefore, for every ordered shadow micro-action,
\[
  g+\Delta\Phi_w\le f,
\]
where $f=w_p$ is its shadow fetch cost and $g\in\{0,w_p\}$ is the physical
fetch cost.  Batch-time repairs satisfy the same inequality with right-hand
side zero.  The caches agree initially and the final potential is
nonnegative, so telescoping gives $M_w(L)\le M_w(O)$.

The delay comparison from Lemma~\ref{lem:lazy-optimum} is unchanged.  Every
positive-length physical $p$-episode starts no earlier and ends at the same
shadow load as a distinct shadow $p$-episode.  Its length is therefore no
larger, while zero-length batch repairs add no delay.  The construction uses
only the current batch, the current shadow action, and the two current
caches, so it is causal.
The arrival-first, same-time, and terminal cases are exactly those recorded
in Remark~\ref{rem:projection-boundaries}; weights change only the movement
potential.
\end{proof}

The exact dynamic programs also extend by changing only the service-edge
cost.

\begin{theorem}[Exact weighted fixed-number-of-holes algorithms]
\label{thm:weighted-fixed-hole}
The weighted offline optimum can be computed in $O(nk)$ time and $O(k)$
working space when $m=k+1$.  For general $r=m-k$, it can be computed in
\[
  O\!\left(n\binom mr2^r(1+rk)\right)
\]
arithmetic operations and $O\!\left(\binom mr2^r\right)$ numeric-label
working space in the unit-cost indexed-state model of
Theorem~\ref{thm:general-dp}.  Thus the weighted problem is polynomial-time
solvable for every fixed $r$; an explicit representation incurs the same set
and indexing overhead described there.
\end{theorem}

\begin{proof}
The holding increment over a gap of length $\Delta$ remains $|B|\Delta$.
With one hole $h$, serving its pending episode and creating hole $v\ne h$
costs $w_h$.  In the intercept implementation of
Theorem~\ref{thm:one-hole-dp}, the service update becomes
\[
  U_v\leftarrow
  \min\left\{U_v,\ t_i+\min_{h\ne v}(P_h+w_h)\right\}.
\]
The minimum and second minimum of the values $P_h+w_h$ evaluate all
exclusions in $O(k)$ time per arrival layer.

For general $r$, the state is still $(H,B)$.  A service of $p\in B$ that
evicts $v\notin H$ follows the same edge
\[
  (H,B)\longrightarrow (H-p+v,B-p),
\]
now with cost $w_p$.  Each edge decreases $|B|$, so the same zero-time DAG
and the same state and edge counts apply.  Lemma~\ref{lem:weighted-lazy}
shows that these lazy paths contain an optimum.  For rational input data,
the additions and comparisons have polynomial bit complexity.
\end{proof}

The reachability conditions of Proposition~\ref{prop:same-epoch-closure}
are unchanged.  Moreover, every action word joining $(H,B)$ to $(H',B')$
loads exactly the pages in $B\setminus B'$, and hence has weighted movement
$\sum_{p\in B\setminus B'}w_p$, independent of its micro-order.

\subsection{A spread-sensitive temporal reduction}

Fix $\theta>0$.  The first raw arrival to $p$ not covered by an open window
anchors the page-specific window
\[
  [a_p,d_p]=[a_p,a_p+\theta w_p].
\]
As in Section~\ref{sec:online-general-randomized}, all later $p$-arrivals
through the right endpoint join the window, and one virtual request $p@d_p$
is output after the complete raw batch at $d_p$.  Simultaneous virtual
requests use a fixed page order.  The resulting sequence is determined by
the timed input and the weights, before any cache choices or random bits.

Run a lazy classical weighted-paging algorithm $A$ on this virtual sequence
in a shadow cache.  At a representative $p@d_p$, perform a physical action
only when $p$ is currently pending; then load $p$ and evict a fixed page in
$C\setminus S$, exactly as in Algorithm~\ref{alg:aw}.  Let $F_w$ denote the
total shadow fetch cost and let $G_w$ denote the total physical fetch cost.

\begin{lemma}[Weighted pending-only projection]
\label{lem:weighted-physical}
Pathwise,
\[
  G_w\le F_w,
  \qquad
  \ALG_w\le(1+\theta)G_w\le(1+\theta)F_w.
\]
Every physical replacement is nonproactive.
\end{lemma}

\begin{proof}
Use $\Phi_w(S,C)=\sum_{p\in S\setminus C}w_p$.  A shadow fault that fetches
$p$ costs $w_p$ and increases $\Phi_w$ by at most $w_p$.  If $p$ is
physically pending after the shadow request, then $p\in S\setminus C$; the
physical fetch of $p$ costs $w_p$ and decreases $\Phi_w$ by exactly $w_p$.
A shadow hit followed by such a repair has physical cost plus potential
change zero.  Thus physical cost plus potential change is at most the shadow
fault cost at every representative.  Telescoping gives $G_w\le F_w$.

No earlier $p$-representative lies inside the current $p$-window, and
physical actions load only their own representative pages.  Hence a physical
$p$-episode created in the window remains pending until $d_p$ and has age at
most $\theta w_p$.  Its fetch and delay together cost at most
$(1+\theta)w_p$.  Summing over the physical replacements proves the second
inequality.  Remark~\ref{rem:aggregation-boundaries} applies verbatim to
page-dependent window lengths.
\end{proof}

We next compare the virtual sequence with a delayed schedule $O$.  Let $M_w$
and $D$ denote its weighted movement and delay.  A representative $p@d_p$ is
\emph{bad} if $p\notin C_O^+(d_p)$.  Let $Z$ denote the number of bad
representatives and let
\[
  Z_w=\sum_{p@d_p\ \mathrm{bad}}w_p.
\]

\begin{lemma}[Weighted bad-representative charges]
\label{lem:weighted-bad}
For every feasible delayed schedule,
\[
  Z_w\le \frac D\theta+\rho M_w,
  \qquad
  Z\le \frac{D}{\theta w_{\min}}+\frac{M_w}{w_{\min}}.
\]
\end{lemma}

\begin{proof}
Consider a bad $p$-window.  If $p$ is absent when the anchor batch arrives
and never enters the cache through the post-action state at its deadline,
then its anchor request remains pending for $\theta w_p$ time.  Charge the
window to this holding segment.  Such segments are page-labelled and
disjoint for successive windows, so the total weight of these bad
representatives is at most $D/\theta$.  Since each weight is at least
$w_{\min}$, their number is at most $D/(\theta w_{\min})$.

Otherwise, the window contains a last transition that evicts $p$ before the
post-action deadline state.  Charge the representative to that eviction.
Successive windows of one page have disjoint interiors, so each replacement
receives at most one such charge.  The number of replacements is at most
$M_w/w_{\min}$.  Each charged representative has weight at most $w_{\max}$,
so their total weight is at most
$w_{\max}M_w/w_{\min}=\rho M_w$.  Adding the two classes proves both
inequalities.
\end{proof}

This comparison of different page identities is the only source of $\rho$ in
the reduction.  The weighted potential and the page-scaled window are exact;
spread enters when the weight of a bad representative or a restored detour
page is compared with the fetch cost of a different page in the delayed
comparator.

\begin{lemma}[Weighted virtual comparator]
\label{lem:weighted-comparator}
Let $\OPT_{\mathrm{wpg}}(\sigma_{\theta,w})$ denote the classical weighted-
paging optimum for the virtual sequence and the common initial cache.  Then
\[
  \OPT_{\mathrm{wpg}}(\sigma_{\theta,w})
  \le
  \frac{1+\rho}{\theta}D+(1+2\rho)M_w.
\]
\end{lemma}

\begin{proof}
Construct a proactive classical weighted-paging schedule that mimics every
replacement of $O$, at total fetch cost $M_w$.  A good representative hits.
At a bad representative for $p$, fetch $p$ while evicting an arbitrary
cached page $q$, serve $p$, and immediately fetch $q$ back while evicting
$p$.  This detour restores $C_O^+(d_p)$ and costs $w_p+w_q$.  The first
fetches over all bad representatives cost $Z_w$, while the restoring fetches
cost at most $w_{\max}Z$.  Lemma~\ref{lem:weighted-bad} therefore bounds the
proactive cost by
\begin{align*}
  M_w+Z_w+w_{\max}Z
  &\le M_w+\frac D\theta+\rho M_w
     +\frac{\rho D}{\theta}+\rho M_w \\
  &=\frac{1+\rho}{\theta}D+(1+2\rho)M_w.
\end{align*}

For completeness, this proactive schedule can be projected to a lazy
classical one without increasing fetch cost.  Couple its cache $A$ to a lazy
cache $L$ and use
\[
  \Psi(A,L)=\sum_{x\in A\setminus L}w_x.
\]
A proactive fetch of $x$ costs $w_x$ and increases $\Psi$ by at most $w_x$.
When a virtual request finds $p\in A\setminus L$, equal cache sizes provide
some $q\in L\setminus A$; the lazy cache fetches $p$, pays $w_p$, and lowers
$\Psi$ by exactly $w_p$.  The caches agree initially and the final cache is
free, so telescoping shows that the lazy cost is no larger.  The classical
optimum is therefore bounded by the displayed proactive cost.
\end{proof}

Combining the two lemmas gives a black-box reduction.

\begin{theorem}[Weighted aggregation reduction]
\label{thm:weighted-reduction}
Suppose that a lazy classical weighted-paging algorithm $A$ satisfies
\[
  \E[F_A(\sigma)]
  \le \alpha\OPT_{\mathrm{wpg}}(\sigma)+\beta
\]
for every fixed request sequence with the common initial cache, where
$F_A(\sigma)$ is its total fetch cost.  Its
page-weighted temporal projection satisfies, against every delayed comparator
with costs $M_w,D$,
\[
  \E[\ALG_w]
  \le
  \alpha(1+\theta)
  \left(
    \frac{1+\rho}{\theta}D+(1+2\rho)M_w
  \right)
  +(1+\theta)\beta.
\]
For randomized $A$, the raw timed input is fixed obliviously.
\end{theorem}

\begin{proof}
The page-specific windows are a deterministic function of the input and the
known weights.  Thus the classical guarantee applies to the virtual sequence
without conditioning on the shadow cache.  Apply
Lemmas~\ref{lem:weighted-physical} and
\ref{lem:weighted-comparator}.
\end{proof}

The two multiplicative coefficients balance at
\[
  \theta=\frac{1+\rho}{1+2\rho},
\]
where both equal $3\rho+2$.

\begin{corollary}[Weighted offline and online guarantees]
\label{cor:weighted-guarantees}
For every finite weighted page universe:
\begin{enumerate}[label=\textup{(\roman*)},leftmargin=*]
\item there is a deterministic nonproactive polynomial-time
      $(3\rho+2)$-approximation offline;
\item there is a deterministic online guarantee with multiplicative factor
      $O(\rho k)$;
\item there is a randomized online guarantee with multiplicative factor
      $O(\rho H_k)$ against an oblivious adversary.
\end{enumerate}
Any additive term in the online bounds is the term inherited from the
classical weighted-paging algorithm, multiplied by $1+\theta$.
\end{corollary}

\begin{proof}
Offline, compute an exact classical weighted-paging optimum for the virtual
sequence by its standard polynomial-time min-cost flow formulation
\cite{fairstein2024distributional} and use
$\alpha=1$, $\beta=0$.  Online, use a deterministic weighted-paging algorithm
with multiplicative factor $k$.  For $k=1$, the unique lazy paging rule is
$1$-competitive.  For $k\ge2$, use the randomized
$O(\log k)=O(H_k)$ algorithm of Bansal, Buchbinder, and Naor
\cite{young1994kserver,bansal2007weighted}.  Theorem~\ref{thm:weighted-reduction}
and the displayed choice of $\theta$ give the claims.
\end{proof}

When $\rho=1$, the comparison coefficient $3\rho+2$ is five, exactly the
coefficient in the unit-cost temporal reduction.  Thus bounded weight spread
preserves the unit-cost offline constant and the deterministic and randomized
online orders.  The next subsection shows why arbitrary spread is different.

\subsection{A temporal lower bound with one cache slot}

The strict weight dependence is not only a weakness of the upper-bound proof.  A
cheap pending page can be served now, at the risk of displacing an expensive
resident page, or can wait for a possible future hit to that resident page.
The following two-input construction makes this uncertainty quantitative.

\begin{theorem}[Strict weight-spread lower bound]
\label{thm:weighted-lower}
Fix $u>0$ and $\rho\ge2$.  There is a two-page weighted instance with cache
size one, minimum weight $u$, and weight spread $\rho$ such that every
randomized online algorithm has an
obliviously fixed input $I$ satisfying
\[
  \frac{\E[\ALG_w(I)]}{\OPT_w(I)}
  \ge \frac{\sqrt\rho-1/\sqrt\rho}{2}
  \ge \frac{\sqrt\rho}{4}.
\]
Hence no strict multiplicative competitive ratio depending only on $k$ is
possible over arbitrary positive page weights.
\end{theorem}

\begin{proof}
Let the initial cache contain page $q$, with $w_q=\rho u$, and set $w_p=u$.
Put $\lambda=\sqrt\rho$ and
$T=u\lambda=\sqrt{w_pw_q}$.  Consider the input distribution
\[
  I_0=(p@0)
  \quad\text{with probability }1-\frac1\lambda,
  \qquad
  I_1=(p@0,q@T)
  \quad\text{with probability }\frac1\lambda.
\]
Fix any deterministic algorithm.  The two timed histories are identical
strictly before $T$.  If the algorithm loads $p$ at some time $\tau<T$, then on
$I_1$ it must eventually fetch $q$ at cost $\rho u$; proactively restoring $q$
before its request already pays the same cost.  Its expected cost under the
distribution is therefore at least $\rho u/\lambda=u\lambda$.

Otherwise, on $I_0$ the request to $p$ remains pending through time $T$ and
has accumulated delay $T$ before its eventual fetch cost $u$.  This
contributes expected cost at least
\[
  \left(1-\frac1\lambda\right)(u+T)
  =u\left(\lambda-\frac1\lambda\right).
\]
An action exactly at $T$ does not create a gap: on $I_1$
the complete $q$-batch is processed first, while on $I_0$ the $p$-episode
has still accumulated delay $T$.

The offline optimum loads $p$ immediately on $I_0$, so $\OPT_w(I_0)=u$.
On $I_1$, it keeps $q$ through its request at $T$ and then loads $p$, giving
$\OPT_w(I_1)=T+u$.  Hence the expected optimum is
\[
  \left(1-\frac1\lambda\right)u
  +\frac1\lambda(u+u\lambda)=2u,
\]
whereas every deterministic algorithm has expected cost at least
$u(\lambda-1/\lambda)$.
The same bound holds after averaging over the fixed coin outcomes of any
randomized algorithm.  Therefore one of the two inputs, fixed independently
of those coins, has ratio at least
$(\lambda-1/\lambda)/2$.  The final inequality in the
statement uses $\rho\ge2$.
\end{proof}

Theorem~\ref{thm:weighted-lower} concerns the strict convention fixed in
Section~\ref{sec:model}.  If an additive constant may depend on the complete
fixed weight system, one copy of the construction can be absorbed into that
constant; the theorem does not claim the same lower bound under that
different convention.  Separately, the classical unit-weight instances
still give the $\Omega(k)$ deterministic and $\Omega(H_k)$ randomized
orders.  We do not claim a product of the spatial and weight-spread lower
bounds.

The construction also explains why the direct rule ``serve $p$ when its age
reaches $\theta w_p$'' is not weight-robust.  With $w_p=1$, it serves $p$ at
constant relative age even when doing so removes a resident page whose weight
is $\rho$ times larger.
A future algorithm with better dependence on $\rho$ must therefore combine
temporal thresholds with the cost of the resident pages at risk, rather than
assigning a timer to the pending page alone.

\section{Discussion: Survival and Failure}
\label{sec:conclusion}

The guiding question of this paper was what remains of paging after a miss no
longer forces an immediate replacement.  The answer is not uniform.  Maximum
delay preserves paging's competitive hierarchy more robustly than its
classical offline optimality structure under unit fetch costs.  Unequal fetch
costs preserve much of the same structure, but reveal a second boundary:
timing uncertainty can depend on the weight spread even with one cache slot.

\paragraph{What survives.}
Under unit fetch costs, the deterministic online order remains $\Theta(k)$,
and the randomized order remains $\Theta(H_k)$ against an oblivious adversary.
LRU remains a useful eviction rule after a timer supplies the
missing timing decision.  Marker remains useful after cache-independent
aggregation supplies an ordinary-paging sequence fixed by the input.  Belady also
remains exact inside that virtual ordinary-paging problem.  Finally, the
classical preference for lazy physical behavior survives particularly
strongly: the causal lazy projection turns every schedule into a
nonproactive one without increasing movement or delay.  Its weighted form,
as well as the exact one-hole and fixed-number-of-holes dynamic programs, uses
the same
state structure.  When the weight spread is bounded, the offline constant
approximation and the deterministic and randomized online orders also
remain bounded as in the unit-cost setting.

\paragraph{What changes.}
Immediate service is no longer without loss.  A direct delayed Marker phase
can depend on earlier random victims, so its request stream need not remain
oblivious.  Most sharply, Belady's farthest-next-use rule can fail for the
physical delayed optimum already with three pages.  Waiting couples the time
of service, the pending page being loaded, and the victim being evicted.  The
exact offline state therefore retains both a hole configuration and a pending
subset beyond the one-hole case.  This does not prove computational hardness;
the unrestricted exact complexity remains unresolved.

Unequal fetch costs change the online picture more sharply than they do in
classical weighted paging.  Serving a cheap pending page may displace an
expensive resident page just before its next request; postponing that service
may instead waste delay if the resident page never returns.  The
$\Omega(\sqrt\rho)$ construction of Section~\ref{sec:weighted} isolates this
choice with two pages and one cache slot.  Thus strict weight dependence is
not only a consequence of choosing among many victims.

The common explanation is
\[
  \text{repeated requests do not add urgency, but cache state remains persistent.}
\]
The holding identity $\Delay=\int|\Pending(t)|\,dt$ removes repeated requests
to one missing page from the instantaneous urgency count.  It does not remove
the cache configuration from which later decisions begin.  In static
maximum-delay aggregation, combining repeated requests can support a block
dynamic program; here the persistent cache remains.

The temporal reduction makes the distinction algorithmic: input-only windows
remove cache dependence from the timing layer, a classical algorithm acts in
a shadow cache, and a discrepancy potential restores physical feasibility.
This one architecture gives the unit-cost approximation and randomized
online bounds and extends to weighted windows.  Belady is used only in the
virtual instance, so its role there is consistent with its failure for the
physical delayed optimum.

We leave four main open questions.
\begin{enumerate}[leftmargin=*]
\item Determine the exact complexity when $r=m-k$ is part of the input.  The
      present results give an exact fixed-$r$ configuration DP and a
      polynomial approximation, but neither a polynomial exact algorithm nor
      a hardness classification.
\item Improve the unit-cost factor five and the weighted
      $(3\rho+2)$ factor, and determine the integrality gap of the marginal
      LP.  In particular, it remains possible that a different weighted
      comparator gives an approximation independent of $\rho$.
\item Close the online weight-spread gap.  The present strict lower bound is
      $\Omega(\sqrt\rho)$ for $k=1$, while the temporal reduction gives
      $O(\rho)$ at that cache size.  Even the one-slot case calls for a
      victim-aware timing rule.
\item Determine whether a shadow cache is necessary for a nonproactive
      $O(H_k)$ algorithm, or whether the same input-determined phase
      structure has a direct paging-native implementation.
\end{enumerate}

\paragraph{General-metric outlook.}
Paging is the uniform-metric special case of the $k$-server problem, but the
general metric makes the service scale depend on both the request point and
the current server configuration.  Appendix~\ref{sec:kserver-transfer} is
deliberately an outlook rather than a solution of the general delayed
$k$-server problem.  It records a first transfer for the direct endpoint-batch
extension.  If
$\delta_X$ and $\Delta_X$ denote the smallest nonzero distance and the
diameter of a finite metric, respectively, then any lazy classical $k$-server
algorithm with a fixed-sequence $(\alpha,\beta)$ guarantee from the common
initial configuration yields a delayed algorithm with multiplicative factor
$\alpha(1+4\Delta_X/\delta_X)$ and an additive term of at most
$(1+2\Delta_X/(\delta_X+2\Delta_X))\beta$.  For a randomized classical
algorithm, the raw timed input must be fixed by an oblivious adversary.  This
establishes competitiveness on every fixed finite metric, but a single-scale
version of the reduction cannot remove the
aspect ratio.  An aspect-independent result appears to require multiscale
timing together with a persistent, geometry-aware assignment of pending
requests to servers.  If several server moves instead share one maximum-delay
charge, the configuration-space requests become sets rather than points and
even this transfer no longer applies.  We view these two general-metric
questions as distinct directions beyond paging.

Taken together, these results separate what waiting changes from what it
preserves.  Allowing misses to wait leaves the classical competitive orders
intact, but changes the mechanisms supporting them and breaks the standard
physical offline victim selection.  With unequal fetch costs, the state being
protected also has a price, and the decision to wait becomes spread-sensitive.
Per-replacement maximum delay therefore groups repeated requests in the
urgency count without removing either the persistent cache state or the cost
of displacing it.

\appendix
\section{Online Algorithms with One Hole}
\label{sec:online-one-hole}

This section assumes $m=k+1$.  The unique hole can have at most one pending
page episode, which makes the busy intervals of a fixed-threshold algorithm
strictly disjoint.

\subsection{Deterministic threshold LRU}

Fix a total order $\prec$ on $\Pages$.  Initially cached pages have artificial
last-arrival keys earlier than every real arrival, with ties resolved by
$\prec$.  A batch at time $t$ assigns every requested page last-arrival time
$t$, again resolving equal-time keys by $\prec$.  Loading a page creates no
additional arrival touch, and every LRU victim tie is resolved by $\prec$.

Fix $\theta>0$.  When the hole is first requested, wait exactly $\theta$,
then load it and evict the least recently requested cached page.  A complete
batch at the deadline is processed before the replacement.

\begin{theorem}[One-hole threshold-LRU bound]
\label{thm:one-hole-lru}
Let $M$, $D$, and $N$ denote the movement cost, delay cost, and number of
online services, respectively.  The first two quantities refer to a lazy
offline optimum.  Then
\[
  N\le \frac{D}{\theta}+(k+2)M.
\]
Consequently,
\[
  \ALG\le
  (1+\theta)\max\left\{\frac1\theta,k+2\right\}\OPT.
\]
Taking $\theta=1/(k+2)$ gives a strict multiplicative $(k+3)$ bound.
\end{theorem}

\begin{proof}
Each online busy interval has the form $B=[a,a+\theta]$.  Consecutive online
activation times are separated by strictly more than $\theta$: after service
at $a+\theta$, the new hole is clean, and its next activation can occur only
at a strictly later arrival epoch.

Call $B$ \emph{covered} if $B\cap\mathcal J_e\ne\varnothing$ for some OPT
episode $e$.  Fix such an episode of length $L=s_e-a_e$.  If $h$ busy
intervals meet $\mathcal J_e$, their left endpoints lie in an interval of
length $L+\theta$, while consecutive left endpoints differ by strictly more
than $\theta$.  Hence $h\le L/\theta+2$, including when $L=0$ and $\mathcal
J_e$ is a singleton.  Summing over the $M$ lazy-OPT episodes, and allowing an
online interval meeting several spans to be overcounted, gives at most
$D/\theta+2M$ covered intervals.

Consider an uncovered busy interval.  Every lazy OPT movement clears one
episode and occurs in its closed service span, while every request to OPT's
current hole lies in the span of the episode that it starts or joins.  The
busy interval therefore contains neither an OPT movement nor a request to
OPT's hole.  It lies in one OPT idle component with a fixed hole $g$, and no
request to $g$ occurs there.  At each online service in the component, either
$g$ is evicted or its LRU rank among the online cached pages decreases: the
loaded page has an actual request inside the component and is strictly more
recent than $g$.  Once $g$ is the online hole, no further uncovered interval
can start before OPT moves again.  Thus every noninitial idle component has at
most $k$ uncovered services.  The initial component has none because the two
caches, and hence their holes, agree initially; OPT has at most $M$
noninitial components.  Therefore
$N\le D/\theta+(k+2)M$.  Since every online service costs exactly
$1+\theta$, the competitive bound follows.
\end{proof}

\subsection{Randomized service-anchored Marker}

Classical arrival-defined Marker phases do not survive unchanged: a pending
hole can cross an arrival phase boundary, and at its service time every cached
page may already be marked.  The following service-anchored rule avoids that
defect.

Initially mark all $k$ cached pages.  Every page arrival marks that page.  A
pending hole is served when its age reaches one.  If an unmarked cached page
exists, evict one uniformly.  Otherwise reset the marks to contain only the
page being loaded, and evict a uniformly random cached page.  The newly
created hole is clean and unmarked.

\begin{theorem}[One-hole service-anchored Marker]
\label{thm:sam}
Against an oblivious adversary,
\[
  \E[\ALG]\le 6H_k\OPT+2H_k+2,
\]
where $H_k=1+1/2+\cdots+1/k$.
\end{theorem}

\begin{proof}
If no service occurs, the claim is immediate.  Assume otherwise and set
$H_0=0$.
Immediately before the fresh victim coin of a reset is revealed, exactly the
loaded page is marked and the new hole is uniform among the $k$ unmarked
pages.  For $u\ge0$, let $F(u)$ denote the
supremum conditional expected number of services remaining before the next
reset, including that reset service, over admissible histories and oblivious
input suffixes at a clean-hole state with no active pending episode, in which
$u$ pages are unmarked and the hole is uniform among them under information
that has not revealed that fresh hole coin.  Put $F(0)=0$.

Condition on such a history with $u>0$.  Ignore batches that touch no new
unmarked page, and suppose the next relevant batch touches $g$ of the $u$
pages; if there is no such batch, the remaining service count is zero.  For
analysis, defer revealing the uniform hole coin until this batch distinguishes
it.  With probability $(u-g)/u$ the hole is untouched; conditional on this
event it remains uniform among the $u-g$ unmarked pages.  With probability
$g/u$ the batch activates the hole and exactly one service occurs one time
unit later.  Additional arrivals before that service may mark more pages, so
at service time the number $u'$ of unmarked victim candidates need only
satisfy $u'\le u-g$.  If $u'=0$, that service is the next reset.  Otherwise
its fresh uniform victim coin makes the new hole uniform among those $u'$
pages.  Strong induction on $u$ and monotonicity of harmonic numbers therefore give
\[
  F(u)
  \le \frac{u-g}{u}H_{u-g}
     +\frac gu\bigl(1+H_{u-g}\bigr)
  =H_{u-g}+\frac gu
  \le H_u.
\]
The last inequality follows because each of the $g$ final summands of $H_u$
is at least $1/u$.  Hence the expected number of services after one reset up
to and including the next reset, or up to the end of the input if no reset
occurs, is at most $H_k$.  The conditioning is with respect to the information
immediately before the fresh victim coin is drawn.  Conditional on the
realized post-coin cache, the hole is known and the bound $F(u)\le H_u$ need
not hold.  Obliviousness fixes the future timed input suffix under the
required pre-coin conditioning.

Let $s_j$ and $s_{j+1}$ denote consecutive reset times, and let $a_j$ denote
$s_j-1$, the activation time of the reset service at $s_j$.  Resetting
retains only the loaded page's mark; that page was requested at $a_j$.
Every other page is then unmarked, and the next reset implies that each is
requested by time $s_{j+1}$.  Thus the interval
\[
  I_j=[a_j,s_{j+1}+1]
\]
contains requests to $k+1$ distinct pages.  A fixed $k$-cache misses at least
one of them.  Consider an arbitrary comparator.  If it moves during $I_j$,
the interval already contains one unit of movement.  Otherwise its cache is
fixed, so one of the $k+1$ requested pages is absent throughout.  That page is
requested no later than the reset time $s_{j+1}$, one full time unit before
the right endpoint of $I_j$.  Serving it requires movement, whereas leaving
it pending contributes at least one unit of holding inside $I_j$.  Thus every
comparator pays at least one unit on $I_j$.

With one hole and the arrival-first convention, consecutive online services
are separated by strictly more than one time unit: after a service the new
hole is clean, its next activation is strictly later, and its deadline is one
unit after that activation.  The intervals $I_j$ consequently have overlap
at most three.  If $R$ is the number of completed reset-to-reset blocks, then
$R\le3\OPT$ pathwise.

The first online service is a reset because all initial cache pages are
marked; count it separately.  Index reset services, including the first, by
$j=0,1,\ldots$.  Let $E_j$ be the event that reset $j$ occurs, and let
$\mathcal F_{j^-}$ contain the raw input prefix, all earlier victim coins, the
current marks, and the fact that this service is a reset, but not its fresh
victim coin.  Thus $E_j\in\mathcal F_{j^-}$.

On $E_j$, let $B_j$ count the subsequent services up to and including the
next reset, or all subsequent services if no later reset occurs; put $B_j=0$
on $E_j^c$.  The bound $F(k)\le H_k$, applied with the pre-coin conditioning,
gives
\[
  \E[B_j\mid\mathcal F_{j^-}]\le H_k
  \quad\text{on }E_j,
  \qquad
  \E[B_j]\le H_k\Pr(E_j).
\]
The finite input induces only finitely many services.  There are exactly
$R+1$ reset events, and every service after the first belongs to exactly one
$B_j$.  Therefore
\[
  \E[N]
  =1+\sum_{j\ge0}\E[B_j]
  \le1+H_k\sum_{j\ge0}\Pr(E_j)
  =1+H_k\E[R+1]
  \le 1+H_k(3\OPT+1).
\]
Every threshold-one service costs exactly two, which gives the displayed
bound.
\end{proof}

The oblivious-adversary restriction is essential to this proof.  An adaptive
adversary that observes the random victim may request it next, as in classical
randomized paging.

\section{The Marginal LP and Open Approximation Questions}
\label{sec:approximation}

The exact general-hole dynamic program is exponential when the effective
number of holes is not fixed, whereas
Theorem~\ref{thm:offline-five-approx} gives a polynomial-time
$5$-approximation.  The following relaxation is therefore not needed to
establish constant approximability.  It remains relevant to improving the
factor, understanding the polyhedral structure, and separating approximation
from exact complexity.  State explosion alone neither proves NP-hardness nor
excludes a polynomial exact algorithm.

\subsection{A marginal linear program}

After each arrival epoch introduce $r$ ordered zero-time service slots.  For
$i=1,\ldots,n$, let $x^{0,-}_{p,i},x^{1,-}_{p,i}$ be the clean- and
pending-hole marginals immediately before batch $A_i$.  Let
$x^0_{p,i,0},x^1_{p,i,0}$ denote the post-batch marginals before the first
slot, and let $x^0_{p,i,\ell},x^1_{p,i,\ell}$ be the marginals after slot
$\ell\in\{1,\ldots,r\}$.  At every one of these stages impose
\[
  x^0_{p,i,\ell},x^1_{p,i,\ell}\ge0,
  \qquad
  x^0_{p,i,\ell}+x^1_{p,i,\ell}\le1,
  \qquad
  \sum_p(x^0_{p,i,\ell}+x^1_{p,i,\ell})=r,
\]
and impose the analogous constraints on pre-batch marginals.

The initial and cross-epoch links are
\[
  x^{0,-}_{p,1}=\mathbf 1[p\notin C_0],
  \qquad x^{1,-}_{p,1}=0,
\]
and, for $i\ge2$ and $j\in\{0,1\}$,
\[
  x^{j,-}_{p,i}=x^j_{p,i-1,r}.
\]
The complete arrival batch acts linearly.  If $p\in A_i$, then
\[
  x^0_{p,i,0}=0,
  \qquad
  x^1_{p,i,0}=x^{0,-}_{p,i}+x^{1,-}_{p,i};
\]
if $p\notin A_i$, then
\[
  x^0_{p,i,0}=x^{0,-}_{p,i},
  \qquad
  x^1_{p,i,0}=x^{1,-}_{p,i}.
\]

For $p\ne q$, let $y_{pq,i,\ell}\ge0$ be the fractional movement in slot
$(i,\ell)$ from pending hole $p$ to cached victim $q$.  Write
\[
  o_{p,i,\ell}=\sum_{q\ne p}y_{pq,i,\ell},
  \qquad
  \operatorname{in}_{q,i,\ell}=\sum_{p\ne q}y_{pq,i,\ell}.
\]
Relative to the immediately preceding slot-state, impose
\[
  o_{p,i,\ell}\le x^1_{p,i,\ell-1},
  \qquad
  \operatorname{in}_{q,i,\ell}
  \le1-x^0_{q,i,\ell-1}-x^1_{q,i,\ell-1},
  \qquad
  \sum_{\substack{p,q\in\Pages\\p\ne q}}y_{pq,i,\ell}\le1,
\]
and update
\[
  x^1_{p,i,\ell}=x^1_{p,i,\ell-1}-o_{p,i,\ell},
  \qquad
  x^0_{p,i,\ell}=x^0_{p,i,\ell-1}
    +\operatorname{in}_{p,i,\ell}.
\]
The objective is
\[
  \min\quad
  \sum_{i=1}^n\sum_{\ell=1}^r
  \sum_{\substack{p,q\in\Pages\\p\ne q}}y_{pq,i,\ell}
  +\sum_{i=1}^{n-1}(t_{i+1}-t_i)\sum_p x^1_{p,i,r},
\]
with terminal constraints $x^1_{p,n,r}=0$ for every page $p$.
Every integral lazy schedule has at most $r$ same-epoch services and embeds
slot by slot in this LP, so $\LP\le\OPT$.  The formulation has
$O(nrm^2)$ variables and constraints and is polynomial in the input size.

\subsection{Certified nonintegrality}

We use the standard supremum convention
\[
  \operatorname{gap}(\LP)
  =\sup_{I:\,\LP(I)>0}\frac{\OPT(I)}{\LP(I)}.
\]

\begin{proposition}[An $8/7$ gap]
\label{prop:gap-eight-seven}
The integrality gap of the marginal LP is at least $8/7$.
\end{proposition}

\begin{proof}
Take four pages $a,b,c,d$, cache size one, initial cache $\{d\}$, and batches
\[
  \{a,d\}@0,
  \qquad \{c\}@\varepsilon,
  \qquad \{b,d\}@1+\varepsilon.
\]
The fractional actions are
\[
  \tfrac12(a\to d),
\]
then
\[
  \tfrac12(c\to a),\quad
  \tfrac12(a\to c),\quad
  \tfrac12(c\to a),
\]
and finally
\[
  \tfrac12(b\to c),\quad
  \tfrac12(d\to b),\quad
  \tfrac12(b\to d).
\]
Table~\ref{tab:lp-gap-trajectory} lists the resulting marginals; a formal sum
such as $b+c+\tfrac12d$ means one unit of clean-hole mass on $b$ and $c$ and
one half unit on $d$.

\begin{table}[h]
  \centering
  \small
  \begin{tabular}{@{}lll c@{}}
    \toprule
    Stage & Clean-hole mass $x^0$ & Pending-hole mass $x^1$
      & Movement increment \\
    \midrule
    after $\{a,d\}@0$ & $b+c$ & $a$ & $0$ \\
    after $\tfrac12(a\to d)$
      & $b+c+\tfrac12d$ & $\tfrac12a$ & $\tfrac12$ \\
    after $\{c\}@\varepsilon$
      & $b+\tfrac12d$ & $\tfrac12a+c$ & $0$ \\
    after $\tfrac12(c\to a)$
      & $\tfrac12a+b+\tfrac12d$ & $\tfrac12a+\tfrac12c$
      & $\tfrac12$ \\
    after $\tfrac12(a\to c)$
      & $\tfrac12a+b+\tfrac12c+\tfrac12d$ & $\tfrac12c$
      & $\tfrac12$ \\
    after $\tfrac12(c\to a)$
      & $a+b+\tfrac12c+\tfrac12d$ & $0$ & $\tfrac12$ \\
    after $\{b,d\}@1+\varepsilon$
      & $a+\tfrac12c$ & $b+\tfrac12d$ & $0$ \\
    after $\tfrac12(b\to c)$
      & $a+c$ & $\tfrac12b+\tfrac12d$ & $\tfrac12$ \\
    after $\tfrac12(d\to b)$
      & $a+\tfrac12b+c$ & $\tfrac12b$ & $\tfrac12$ \\
    after $\tfrac12(b\to d)$
      & $a+\tfrac12b+c+\tfrac12d$ & $0$ & $\tfrac12$ \\
    \bottomrule
  \end{tabular}
  \caption{Marginal trajectory for the $8/7$-gap instance.  Every row has
  total hole mass three and per-page mass at most one.}
  \label{tab:lp-gap-trajectory}
\end{table}

Every displayed half-move respects the outgoing pending mass and the incoming
cached mass in its immediately preceding slot-state.  The movement totals
$\tfrac12+\tfrac32+\tfrac32=7/2$.  Only the half unit of pending $a$-mass
survives a positive time interval, namely $[0,\varepsilon)$, so the holding
cost is $\varepsilon/2$.

By Lemma~\ref{lem:lazy-optimum}, it suffices to consider lazy schedules.
Integral OPT is four.  If both $a,c$ are served before the last epoch, two
movements have occurred and both $b,d$ are then holes, requiring two more.
If exactly one of $a,c$ is served early, the other page together with $b,d$
is pending at the final epoch, again requiring four total movements.  If
neither is served early, their delay plus the three final movements exceeds
four.  Immediate services $a,c,b,d$ attain four.  Thus, for every
$\varepsilon>0$,
\[
  \LP(I_\varepsilon)\le\frac72+\frac\varepsilon2,
  \qquad
  \OPT(I_\varepsilon)=4.
\]
Taking the supremum and then letting $\varepsilon\downarrow0$ gives
\[
  \operatorname{gap}(\LP)
  \ge\sup_{\varepsilon>0}
    \frac{4}{\frac72+\frac\varepsilon2}
  =\frac87.
\]
\end{proof}

Slotwise rounding that preserves the exact clean-hole and pending-hole
marginals is impossible.  For example, take with equal probability the
colored bases $\{A,B\}$ with both pending and $\{C,D\}$ with both pending.
The LP permits half a unit of $A\to B$.  Clearing $A$ must happen in the
first realization, where $B$ is not cached; creating clean $B$ in the second
realization necessarily clears pending $C$ or $D$.  Thus the target colored
marginals are unreachable at any movement cost.

\begin{openproblem}
Can the approximation factor five be improved?  Independently, does the
marginal LP have constant integrality gap under a whole-epoch or amortized
rounding that need not preserve intermediate marginals?
\end{openproblem}

Theorem~\ref{thm:offline-five-approx} does not imply an integrality-gap upper
bound: its analysis compares the returned schedule with $\OPT$, not with
$\LP$.  In particular, no inequality of the form
$\ALG_{\rm off}\le5\LP$ has been proved, and the certified $8/7$ lower bound
remains compatible with either a constant or an unbounded gap.

A complementary direction is to characterize the full epoch-closure
polytope or strengthen the relaxation by configuration-correlation
constraints with a polynomial separation oracle.

\section{A First General-Metric Transfer}
\label{sec:kserver-transfer}

This appendix asks how far the temporal reduction survives beyond the
uniform metric.  We obtain an aspect-ratio-dependent black-box theorem for the
direct endpoint-batch extension of the model.  The result is not presented
as a solution to delayed $k$-server on arbitrary metrics: the final
subsection isolates why a metric-independent theorem needs a different
timing mechanism.

\subsection{Endpoint-batch model}

Let $(X,d)$ denote a known finite metric and let a configuration denote a
multiset of $k$ server locations.  All schedules start from the same
configuration.  A request at an occupied point is served immediately.  If a
point $x$ is unoccupied, all requests to $x$ join one pending episode whose
oldest arrival time is $a_x$.  Moving one server from $y$ to a pending $x$ at
time $t$ costs
\[
  d(y,x)+(t-a_x)
\]
and clears the episode.  A move to a nonpending point is proactive and pays
only distance.  We retain the arrival-first convention and a free terminal
configuration.

This model charges maximum delay within one location episode, but it adds the
charges of distinct episodes.  On a uniform metric, starting from distinct
server locations, every nonproactive trajectory keeps the locations distinct
and coincides with a paging trajectory.  Thus the lazy specialization
recovers the model of Section~\ref{sec:model}.  Assume $|X|\ge2$ and write
\[
  \delta_X=\min_{x\ne y}d(x,y),
  \qquad
  \Delta_X=\max_{x,y}d(x,y),
  \qquad
  \Gamma_X=\frac{\Delta_X}{\delta_X}.
\]
The one-point metric has zero cost and may be omitted.
If $\Pending(t)$ denotes the set of active location episodes at time $t$,
the same holding identity gives
\[
  D=\int_0^\infty |\Pending(t)|\,dt.
\]

For two server configurations $S,C$, let
\[
  \Phi(S,C)
  =\min_{\pi}\sum_{s\in S}d\bigl(s,\pi(s)\bigr)
\]
denote their minimum perfect-matching distance, where the matching is between
server copies.

\begin{theorem}[General-metric temporal transfer]
\label{thm:kserver-transfer}
Suppose a lazy classical $k$-server algorithm $A$ satisfies
\[
  \E[F_A(\sigma)]
  \le \alpha\OPT_{\rm KS}(\sigma)+\beta
\]
for every fixed classical request sequence $\sigma$ with the common initial
configuration, where $F_A(\sigma)$ is its total server movement.  Then the
endpoint-batch delayed model has a causal
nonproactive algorithm satisfying
\[
  \E[\ALG(I)]
  \le
  \alpha(1+4\Gamma_X)\OPT_{\rm del}(I)
  +\left(1+\frac{2\Gamma_X}{1+2\Gamma_X}\right)\beta.
\]
For randomized $A$, the raw timed input is fixed by an oblivious adversary.
\end{theorem}

The theorem is strict when the classical guarantee is strict.  In
particular, the standard work-function guarantee gives an
$O(k\Gamma_X)$ ratio with an additive constant
\cite{koutsoupias1995kserver}.  Its strict $(4k-2)$ form gives a strict
$O(k\Gamma_X)$ ratio \cite{emek2010additive}.  When $X$ is uniform,
$\Gamma_X=1$ and the temporal coefficient is exactly five.

\subsection{Input-only windows and matching projection}

Fix $\theta>0$ and put $L=\theta\delta_X$.  Independently for each point
$x$, the first raw arrival not covered by an open $x$-window anchors
\[
  [a_x,d_x]=[a_x,a_x+L].
\]
All later $x$-arrivals through the right endpoint join the window.  After the
complete raw batch at $d_x$, output one virtual classical request $x@d_x$
and close the window.  Simultaneous representatives use a fixed point order.
The virtual sequence $\sigma_L$ is determined by the timed input before any
server choices or random bits.

Run $A$ on $\sigma_L$ in a shadow configuration $S$ and maintain a physical
configuration $C$, with $S=C$ initially.  At a representative $x@d_x$, first
let $A$ process $x$.  If $x$ is physically pending, take an optimal matching
between $C$ and the updated $S$, choose a physical server matched to a shadow
copy at $x$, and move that server to $x$.  If $x$ is not pending, take no
physical action.

\begin{lemma}[Matching projection]
\label{lem:kserver-matching-projection}
If $F$ and $G$ denote the shadow and physical movement, respectively, then
pathwise
\[
  G\le F,
  \qquad
  \ALG\le(1+\theta)G\le(1+\theta)F.
\]
Every physical movement is nonproactive.
\end{lemma}

\begin{proof}
A shadow move of length $f$ increases $\Phi(S,C)$ by at most $f$, by the
triangle inequality applied to the matching edge incident to the moving
shadow copy.  If a physical repair moves a server from $z$ to $x$, let
$g=d(z,x)$.  The selected matching edge of length $g$ becomes a zero-length
edge, so the repair decreases the minimum matching potential by at least
$g$.  Therefore every ordered representative satisfies
\[
  g+\Phi_{\rm after}-\Phi_{\rm before}\le f.
\]
The configurations agree initially and the terminal potential is
nonnegative.  Summing gives $G\le F$.

Physical actions load only the current representative point.  Hence a
physical episode created in an $x$-window remains pending until its
representative and has age at most $L$.  Because a pending $x$ is unoccupied,
its physical repair has positive length $g\ge\delta_X$.  Thus its delay is at
most $\theta g$.  Summing over repairs proves the displayed cost bound.
Arrival-first processing and the fixed order of simultaneous representatives
give the same argument at shared deadlines.
\end{proof}

\begin{remark}[Multiplicity and event boundaries]
Configurations are multisets.  If several shadow servers occupy $x$, any
copy matched to a physical server may be chosen; if $x$ is physically
pending, no physical copy occupies $x$, so every repair has positive length
at least $\delta_X$.  The arrival-first and fixed-order conventions of
Remark~\ref{rem:aggregation-boundaries} also apply here.  In particular,
same-time departures are read at the post-action cut, and open deadlines
remain scheduled after the final raw arrival without an end-of-input signal.
\end{remark}

\begin{algorithm}[h]
\caption{General-metric temporal projection}
\label{alg:kserver-projection}
\begin{algorithmic}[1]
\State $S\gets C_0$; $C\gets C_0$; initialize no open point windows
\ForAll{raw arrival epochs and window deadlines $t$ in increasing order}
  \State Process the complete raw batch at $t$ against $C$
  \ForAll{raw arrivals $x@t$ not covered by an open $x$-window}
    \State Open $[t,t+L]$ for $x$
  \EndFor
  \State Let $R_t$ denote the representatives whose windows end at $t$
  \ForAll{$x@t\in R_t$ in the fixed point order}
    \State Let the classical algorithm $A$ process $x$ in $S$
    \If{$x$ is physically pending}
      \State Compute a minimum perfect matching between $C$ and $S$
      \State Let $z\in C$ be matched to a shadow server at $x$
      \State Move $z$ to $x$ in $C$ and clear the $x$-episode
    \EndIf
    \State Close the $x$-window
  \EndFor
\EndFor
\end{algorithmic}
\end{algorithm}

\subsection{Comparison with a delayed schedule}

Fix any delayed schedule $O$ with movement $M$ and delay $D$.  Call a
representative $x@d_x$ \emph{bad} if $O$ has no server at $x$ after its
complete raw-batch and action sequence at $d_x$.  Let $Z$ denote the number
of bad representatives.

\begin{lemma}[Bad representatives]
\label{lem:kserver-bad-representatives}
\[
  Z\le \frac{D}{\theta\delta_X}+\frac{M}{\delta_X}.
\]
\end{lemma}

\begin{proof}
Consider a bad $x$-window.  If $O$ has no server at $x$ in any cache
microstate from the anchor batch through the post-action deadline state, the
anchor request remains pending for the full window length
$\theta\delta_X$.  Successive windows of one point have disjoint interiors,
and holding segments at different points are added by the holding identity.
These windows therefore number at most $D/(\theta\delta_X)$.

Otherwise the window contains a last transition that makes $x$ unoccupied.
Charge the representative to that departure.  If another server remains at
$x$, the departure does not qualify.  One movement has only one origin, and
successive windows of one point have disjoint interiors, so every movement is
charged at most once.  Every charged movement has length at least
$\delta_X$, and these windows therefore number at most $M/\delta_X$.
\end{proof}

\begin{lemma}[Virtual $k$-server comparator]
\label{lem:kserver-virtual-comparator}
\[
  \OPT_{\rm KS}(\sigma_L)
  \le
  \frac{2\Gamma_X}{\theta}D
  +(1+2\Gamma_X)M.
\]
\end{lemma}

\begin{proof}
Construct a proactive classical trajectory that mimics all movements of $O$
between representatives.  At a deadline, first mimic all same-time actions
of $O$ and then process the representatives in their fixed order.  A good
representative hits.  For a bad $x@d_x$, move any current server from $q$ to
$x$ and immediately return it to $q$.  The detour serves $x$, restores the
configuration, and costs at most $2\Delta_X$.  Lemma
\ref{lem:kserver-bad-representatives} bounds the proactive cost by
\begin{align*}
  M+2\Delta_X Z
  &\le
  M+2\Delta_X
  \left(\frac{D}{\theta\delta_X}+\frac{M}{\delta_X}\right)\\
  &=\frac{2\Gamma_X}{\theta}D+(1+2\Gamma_X)M.
\end{align*}

For completeness, couple this proactive trajectory to a lazy classical
configuration and use the minimum-matching distance as a potential.  A
proactive move increases the potential by at most its length.  At a virtual
request missing from the lazy configuration, move the physical server
matched to the requested shadow copy; its movement is paid by the potential
drop.  The configurations agree initially and the terminal configuration is
free, so telescoping removes all proactive movements without increasing
cost.  The classical optimum is no larger than the resulting lazy cost.
\end{proof}

\begin{proof}[Proof of Theorem~\ref{thm:kserver-transfer}]
The virtual sequence is an input-only object, so the classical guarantee
applies directly.  Lemmas~\ref{lem:kserver-matching-projection} and
\ref{lem:kserver-virtual-comparator} give
\[
  \E[\ALG]
  \le
  \alpha(1+\theta)
  \left(
    \frac{2\Gamma_X}{\theta}D+(1+2\Gamma_X)M
  \right)
  +(1+\theta)\beta.
\]
The delay and movement coefficients agree when
\[
  \theta=\frac{2\Gamma_X}{1+2\Gamma_X}.
\]
Their common value is $1+4\Gamma_X$.  Apply the inequality to an optimal
delayed schedule.
\end{proof}

\subsection{Why the aspect ratio remains}

The dependence on $\Gamma_X$ is a limitation of the proved single-scale
transfer, not a lower bound for every delayed algorithm.  The following
simple test shows that one common window length cannot give an
aspect-independent shadow reduction.

\begin{figure}[h]
\centering
\begin{tikzpicture}[
  >=Latex,
  point/.style={circle,fill=black,inner sep=1.8pt},
  request/.style={circle,draw=red!70!black,fill=red!12,inner sep=2.3pt},
  note/.style={font=\small,align=left}
]
  \node[note,font=\small\bfseries,anchor=east] at (2.2,1.55) {close alternating test};
  \draw[thick] (2.7,1.35) -- (11.4,1.35);
  \node[point,label=below:$a$] at (2.7,1.35) {};
  \node[request,label=below:$z$] at (4.0,1.35) {};
  \node[point,label=below:$p$] at (11.4,1.35) {};
  \draw[<->,blue!65!black] (2.7,1.82) -- node[above] {$\delta$} (4.0,1.82);
  \node[note,anchor=west] at (4.55,1.55)
    {$z,a,z,a,\ldots$: long-run loss $\ge 1+\tau/\delta$};

  \node[note,font=\small\bfseries,anchor=east] at (2.2,-0.35)
    {far alternating test};
  \draw[thick] (2.7,-0.55) -- (11.4,-0.55);
  \node[point,label=below:$a$] at (2.7,-0.55) {};
  \node[point,label=below:$z$] at (4.0,-0.55) {};
  \node[request,label=below:$p$] at (11.4,-0.55) {};
  \draw[<->,orange!80!black,thick,bend left=18]
    (2.7,-0.38) to node[above] {$p,a,p,a,\ldots$} (11.4,-0.38);
  \node[note,anchor=west] at (4.55,-1.15)
    {$d(a,p)=\Delta$; long-run loss $\ge 1+\Delta/\tau$};
\end{tikzpicture}
\caption{The single-scale trap.  A long window is repeatedly expensive on a
nearby pair, while a short window forces repeated long crossings.  Both tests
are repeatable, so their losses cannot be hidden by an additive constant.}
\label{fig:kserver-two-scale}
\end{figure}
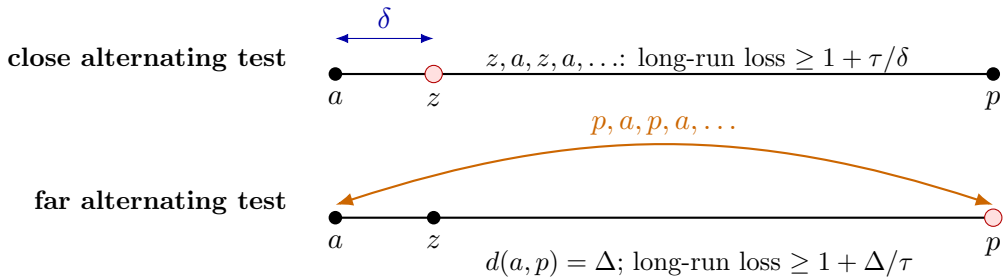

\paragraph{Representative-only single-scale reductions.}
Fix $\tau>0$.  A \emph{representative-only single-scale reduction}
$\mathcal R_\tau$ opens at each point $x$ a window $[t,t+\tau]$ whenever a
raw arrival is not covered by an existing $x$-window.  Later $x$-arrivals
through the right endpoint join that window, and one representative is issued
after the complete batch at its deadline.  The physical server moves only
when processing representatives.  At an $x$-representative it moves to $x$
exactly when $x$ is physically pending, and then clears that episode; if $x$
is not pending it does not move.  Representatives are processed in event
order, with a fixed order at common deadlines.  Internal shadow states are
permitted but do not alter these physical-service rules.

\begin{proposition}[Single-scale representative-only limitation]
\label{prop:kserver-single-scale}
Fix $0<\delta\le\Delta/2$ and let $X=\{a,z,p\}$ be the three-point line metric
with $a=0$, $z=\delta$, and $p=\Delta$, with the common initial server located
at $a$.  Then
\[
  \delta_X=\delta,
  \qquad \Delta_X=\Delta,
  \qquad \Gamma_X=\frac{\Delta}{\delta}.
\]
For every $\tau>0$, if finite constants $c\ge0$ and $\beta$ satisfy
\[
  \mathcal R_\tau(I)\le c\OPT_{\rm del}(I)+\beta
\]
for every finite timed input on this fixed metric, then
\[
  c\ge
  \max\left\{1+\frac{\tau}{\delta},
              1+\frac{\Delta}{\tau}\right\}
  \ge1+\sqrt{\Gamma_X}.
\]
In particular, a representative-only reduction using one common window
length has competitive ratio at least $\sqrt{\Gamma_X}$ even when an additive
constant depending on the fixed metric and on $\tau$ is allowed.
\end{proposition}

\begin{proof}
The nonzero pairwise distances are $\delta$, $\Delta-\delta$, and $\Delta$.
The restriction $\delta\le\Delta/2$ therefore gives the displayed values of
$\delta_X$, $\Delta_X$, and $\Gamma_X$.

Fix $\varepsilon>0$ and release $N$ singleton requests at times
$t_j=j(\tau+\varepsilon)$, for $j=0,\ldots,N-1$.  First alternate their
locations between $z$ and $a$, beginning at $z$.  Every representative occurs
strictly before the next raw arrival, consecutive visits to one point are
more than $\tau$ apart, and induction shows that the server is at the other
point when each representative is processed.  Hence
\[
  \mathcal R_\tau(I_N^{\rm close})=N(\delta+\tau).
\]
Serving every request immediately is feasible and costs $N\delta$, so
\[
  N(\delta+\tau)\le cN\delta+\beta.
\]
Letting $N\to\infty$ gives $c\ge1+\tau/\delta$.  Thus, unlike the former
one-request test, the close-point loss cannot be absorbed by an additive
constant.

For the far-point test use the same release times and alternate between $p$
and $a$, beginning at $p$.  The same event-order argument gives
\[
  \mathcal R_\tau(I_N^{\rm far})=N(\Delta+\tau).
\]
A feasible offline schedule stays at $a$ throughout the raw arrivals.  All
$a$-requests hit, all $p$-requests join the episode anchored at time zero,
and immediately after the final batch the schedule moves once to $p$.  Its
cost is $\Delta+(N-1)(\tau+\varepsilon)$, and hence
\[
  N(\Delta+\tau)
  \le c\bigl(\Delta+(N-1)(\tau+\varepsilon)\bigr)+\beta.
\]
First letting $N\to\infty$ and then $\varepsilon\downarrow0$ gives
$c\ge1+\Delta/\tau$.  Combining the two repeatable tests yields
\[
  c\ge1+\max\left\{\frac{\tau}{\delta},
                    \frac{\Delta}{\tau}\right\}
  \ge1+\sqrt{\frac{\Delta}{\delta}}
  =1+\sqrt{\Gamma_X}.
\]
All representatives precede the next raw arrival strictly, so the
construction uses no simultaneous-event boundary case.  The final offline
move clears the outstanding $p$-episode and is feasible under the free
terminal-configuration convention.
\end{proof}

A purely local assignment rule also fails.  On the line
$a=0,x=1,b=3$ with initial servers at $a,b$, alternate requests between
$x$ and $a$, always releasing the next request just after the preceding
service.  A rule that triggers when the oldest age reaches the nearest-server
distance and then moves that nearest server crosses $a x$ forever.  The
offline solution moves the server at $b$ to $x$ once and thereafter hits
every request.  Thus persistent credit and geometry-aware server assignment
are necessary even when the aspect ratio is constant.

One promising direction is to grow temporal credit at several metric scales
and retain server--subtree commitments across local services.  The missing
global step is a multi-root certificate: one offline server trajectory may
visit terminals belonging to several online forest components, so a direct
component-to-server matching can fail already for $k=2$.

Finally, another natural model lets one or several servers visit multiple
request points in one event and charges one maximum delay over their union.
In configuration space, a request at $x$ then corresponds to the service set
\[
  \{C:C\text{ contains }x\},
\]
not to one configuration.  The preceding request identity no longer pins
the next spatial anchor, and Theorem~\ref{thm:kserver-transfer} does not apply.
This shared-event model should therefore be treated separately from the
endpoint-batch extension proved here.

\section*{Declaration of Generative AI Use}
The authors used OpenAI language-model systems, including Codex, as
interactive assistants during this project.  These systems helped explore
proofs, search for counterexamples, write formal arguments, check
computations, draft LaTeX, and revise the paper.  The human authors chose the
research questions and model, checked the generated arguments, and take full
responsibility for every claim and for the final manuscript.

\bibliographystyle{plain}
\bibliography{references}

@article{fiat1991competitive,
  author  = {Fiat, Amos and Karp, Richard M. and Luby, Michael and McGeoch, Lyle A. and Sleator, Daniel D. and Young, Neal E.},
  title   = {Competitive Paging Algorithms},
  journal = {Journal of Algorithms},
  volume  = {12},
  number  = {4},
  pages   = {685--699},
  year    = {1991},
  doi     = {10.1016/0196-6774(91)90041-V}
}

@article{mcgeoch1991strongly,
  author  = {McGeoch, Lyle A. and Sleator, Daniel D.},
  title   = {A Strongly Competitive Randomized Paging Algorithm},
  journal = {Algorithmica},
  volume  = {6},
  number  = {6},
  pages   = {816--825},
  year    = {1991},
  doi     = {10.1007/BF01759073}
}

@article{belady1966study,
  author  = {Belady, L{\'a}szl{\'o} A.},
  title   = {A Study of Replacement Algorithms for a Virtual-Storage Computer},
  journal = {IBM Systems Journal},
  volume  = {5},
  number  = {2},
  pages   = {78--101},
  year    = {1966},
  doi     = {10.1147/sj.52.0078}
}

@article{sleator1985amortized,
  author  = {Sleator, Daniel D. and Tarjan, Robert E.},
  title   = {Amortized Efficiency of List Update and Paging Rules},
  journal = {Communications of the ACM},
  volume  = {28},
  number  = {2},
  pages   = {202--208},
  year    = {1985},
  doi     = {10.1145/2786.2793}
}

@article{young1994kserver,
  author  = {Young, Neal E.},
  title   = {The $K$-Server Dual and Loose Competitiveness for Paging},
  journal = {Algorithmica},
  volume  = {11},
  number  = {6},
  pages   = {525--541},
  year    = {1994},
  doi     = {10.1007/BF01189992}
}

@inproceedings{bansal2007weighted,
  author    = {Bansal, Nikhil and Buchbinder, Niv and Naor, Joseph},
  title     = {A Primal-Dual Randomized Algorithm for Weighted Paging},
  booktitle = {Proc. FOCS},
  pages     = {507--517},
  year      = {2007},
  doi       = {10.1109/FOCS.2007.43}
}

@inproceedings{fairstein2024distributional,
  author    = {Fairstein, Yaron and Naor, Joseph and Tsachor, Tomer},
  title     = {Distributional Online Weighted Paging with Limited Horizon},
  booktitle = {Proc. APPROX/RANDOM},
  volume    = {317},
  pages     = {15:1--15:15},
  year      = {2024},
  doi       = {10.4230/LIPIcs.APPROX/RANDOM.2024.15}
}

@inproceedings{azar2017osd,
  author    = {Azar, Yossi and Ganesh, Arun and Ge, Rong and Panigrahi, Debmalya},
  title     = {Online Service with Delay},
  booktitle = {Proc. STOC},
  pages     = {551--563},
  year      = {2017},
  doi       = {10.1145/3055399.3055475}
}

@article{gupta2022caching,
  author  = {Gupta, Anupam and Kumar, Amit and Panigrahi, Debmalya},
  title   = {Caching with Time Windows and Delays},
  journal = {SIAM Journal on Computing},
  volume  = {51},
  number  = {4},
  pages   = {975--1017},
  year    = {2022},
  doi     = {10.1137/20M1346286}
}

@inproceedings{krnetic2020kserver,
  author    = {Krneti\'c, Predrag and Melnyk, Darya and Wang, Yuyi and Wattenhofer, Roger},
  title     = {The $k$-Server Problem with Delays on the Uniform Metric Space},
  booktitle = {Proc. ISAAC},
  volume    = {181},
  pages     = {61:1--61:13},
  year      = {2020},
  doi       = {10.4230/LIPIcs.ISAAC.2020.61}
}

@misc{lu2026mlamax,
  author       = {Lu, Tianhang and Ren, Runtian and Liu, Shengcai and Tang, Ke},
  title        = {Online Multi-Level Aggregation with Per-Batch Maximum Delay},
  howpublished = {arXiv:2608.06796},
  year         = {2026},
  eprint       = {2608.06796},
  archivePrefix = {arXiv},
  primaryClass = {cs.DS}
}

@article{asiatici2022request,
  author  = {Asiatici, Mikhail and Ienne, Paolo},
  title   = {Request, Coalesce, Serve, and Forget: Miss-Optimized Memory Systems for Bandwidth-Bound Cache-Unfriendly Applications on {FPGAs}},
  journal = {ACM Transactions on Reconfigurable Technology and Systems},
  volume  = {15},
  number  = {2},
  pages   = {13:1--13:33},
  year    = {2022},
  doi     = {10.1145/3466823}
}

@inproceedings{atre2020delayed,
  author    = {Atre, Nirav and Sherry, Justine and Wang, Weina and Berger, Daniel S.},
  title     = {Caching with Delayed Hits},
  booktitle = {Proc. SIGCOMM},
  pages     = {495--513},
  year      = {2020},
  doi       = {10.1145/3387514.3405883}
}

@inproceedings{gurushankar2025latency,
  author    = {Gurushankar, Keerthana and Singer, Noah G. and Subercaseaux, Bernardo},
  title     = {Latency Guarantees for Caching with Delayed Hits},
  booktitle = {Proc. INFOCOM},
  pages     = {1--10},
  year      = {2025},
  doi       = {10.1109/INFOCOM55648.2025.11044511}
}

@article{feder2004combining,
  author  = {Feder, Tom{\'a}s and Motwani, Rajeev and Panigrahy, Rina and Seiden, Steve and van Stee, Rob and Zhu, An},
  title   = {Combining Request Scheduling with Web Caching},
  journal = {Theoretical Computer Science},
  volume  = {324},
  number  = {2--3},
  pages   = {201--218},
  year    = {2004},
  doi     = {10.1016/j.tcs.2004.05.016}
}

@article{albers2010reordering,
  author  = {Albers, Susanne},
  title   = {New Results on Web Caching with Request Reordering},
  journal = {Algorithmica},
  volume  = {58},
  number  = {2},
  pages   = {461--477},
  year    = {2010},
  doi     = {10.1007/s00453-008-9276-x}
}

@inproceedings{avigdor2013rbm,
  author    = {Avigdor-Elgrabli, Noa and Rabani, Yuval},
  title     = {An Optimal Randomized Online Algorithm for Reordering Buffer Management},
  booktitle = {Proc. FOCS},
  pages     = {1--10},
  year      = {2013},
  doi       = {10.1109/FOCS.2013.9}
}

@inproceedings{adamaszek2011almost,
  author    = {Adamaszek, Anna and Czumaj, Artur and Englert, Matthias and R{\"a}cke, Harald},
  title     = {Almost Tight Bounds for Reordering Buffer Management},
  booktitle = {Proc. STOC},
  pages     = {607--616},
  year      = {2011},
  doi       = {10.1145/1993636.1993717}
}

@article{koutsoupias1995kserver,
  author  = {Koutsoupias, Elias and Papadimitriou, Christos H.},
  title   = {On the $k$-Server Conjecture},
  journal = {Journal of the ACM},
  volume  = {42},
  number  = {5},
  pages   = {971--983},
  year    = {1995},
  doi     = {10.1145/210118.210128}
}

@article{emek2010additive,
  author  = {Emek, Yuval and Fraigniaud, Pierre and Korman, Amos and Ros{\'e}n, Adi},
  title   = {On the Additive Constant of the $k$-Server Work Function Algorithm},
  journal = {Information Processing Letters},
  volume  = {110},
  number  = {24},
  pages   = {1120--1123},
  year    = {2010},
  doi     = {10.1016/j.ipl.2010.09.003}
}

@inproceedings{bhore2026general,
  author        = {Bhore, Sujoy and Paw{\l}owski, Micha{\l} and Umboh, Seeun William},
  title         = {Online {TCP} Acknowledgment under General Delays},
  booktitle = {Proc. APPROX/RANDOM (to appear)},
  year          = {2026}
}

\end{document}